\documentclass[11pt,letterpaper]{article}
\pdfoutput=1

\usepackage[utf8]{inputenc}
\usepackage{lmodern}
\usepackage{amsmath, amsthm, amssymb, mathtools}
\usepackage{amsfonts}
\usepackage{bm}
\usepackage{booktabs}
\usepackage{braket}
\usepackage{xcolor}
\usepackage{enumitem}
\usepackage{etoolbox}
\usepackage{inconsolata}
\usepackage[margin=1in]{geometry}
\usepackage{graphicx}
\usepackage{mleftright}\mleftright
\usepackage{caption}
\usepackage{microtype}
\PassOptionsToPackage{hyphens}{url}
\usepackage{thmtools,thm-restate}
\usepackage{url}
\usepackage[normalem]{ulem}
\usepackage{verbatim}
\usepackage[linesnumbered,ruled,procnumbered]{algorithm2e}
\usepackage[all,cmtip]{xy}

\makeatletter
\newcommand{\para}{%
    \@startsection{paragraph}{4}%
    {\z@}{2ex \@plus 3.3ex \@minus .2ex}{-1em}%
    {\normalfont\normalsize\bfseries}%
}
\makeatother

\DeclarePairedDelimiter{\abs}{\lvert}{\rvert} 
\DeclarePairedDelimiter{\norm}{\lVert}{\rVert} 

\DeclarePairedDelimiter{\parens}{\lparen}{\rparen}

\DeclareMathOperator{\poly}{poly}

\DeclareMathOperator*{\Tr}{Tr}

\DeclareMathOperator{\Aut}{Aut}

\DeclareMathOperator{\Sym}{Sym}

\DeclareMathOperator{\diag}{diag}

\newcommand{\eps}{\varepsilon}

\newcommand{\ZZ}{{\mathbb{Z}}}
\newcommand{\RR}{{\mathbb{R}}}
\newcommand{\CC}{{\mathbb{C}}}

\newcommand{\ii}{\mathrm{i}}

\newcommand{\rd}{{\mathrm{d}}}

\newcommand{\half}{{\tfrac 1 2}}

\newcommand{\calC}{{\mathcal{C}}}
\newcommand{\calD}{{\mathcal{D}}}
\newcommand{\calH}{{\mathcal{H}}}

\newcommand{\calS}{{\mathcal{S}}}
\newcommand{\bigO}[1]{\mathcal{O}\left( #1 \right)}

\newcommand{\one}{{\mathbf{1}}}
\newcommand{\nn}{{\mathsf{n}}} 
\newcommand{\NN}{{\mathsf{N}}} 
\newcommand{\SU}{{\mathsf{SU}}} 
\newcommand{\su}{{\mathfrak{su}}} 
\newcommand{\SO}{{\mathsf{SO}}} 
\newcommand{\so}{{\mathfrak{so}}} 

\newcommand{\pauli}{{\mathbf{P}}} 

\newcommand{\ad}{{\mathsf{ad}}}
\newcommand{\sutwo}{{\mathfrak{su}(2)}} 
\newcommand{\E}{\mathop{\mathbb{E}}}

\newcommand{\skewe}{{\mathbf{E}}}

\usepackage[backend=biber,style=alphabetic,url=false,isbn=false,maxnames=10,minnames=3,maxalphanames=4,minalphanames=3]{biblatex}
\usepackage[colorlinks=true]{hyperref}
\usepackage{xcolor}
\definecolor{darkred}  {rgb}{0.5,0,0}
\definecolor{darkblue} {rgb}{0,0,0.5}
\definecolor{darkgreen}{rgb}{0,0.5,0}
\hypersetup{
  urlcolor   = blue,         
  linkcolor  = darkblue,     
  citecolor  = darkgreen,    
  filecolor   = darkred       
}
\usepackage[nameinlink,capitalize]{cleveref}
\usepackage{tikz}
\usetikzlibrary{quantikz2}

\declaretheorem[numberwithin=section]{theorem}
\declaretheorem[sibling=theorem]{lemma}
\declaretheorem[sibling=theorem]{corollary}
\declaretheorem[sibling=theorem]{proposition}
\theoremstyle{definition}
\declaretheorem[sibling=theorem]{definition}
\declaretheorem[sibling=theorem]{remark}

\begin{document}
\title{
\vspace{-15mm}Efficient approximate unitary designs from random Pauli rotations}

\author{Jeongwan Haah\thanks{Microsoft Quantum, ~~~
$^{**}$Department of EECS, UC Berkeley, ~~~
$^{***}$ Department of Mathematics, MIT.
} ~~~~~~~~ Yunchao Liu$^{**}$ ~~~~~~~~ Xinyu Tan$^{***}$}
\date{}

\maketitle

\begin{quote}
    {\smaller
    \vspace{-5mm}
    We construct random walks on simple Lie groups that quickly converge to the Haar measure for all moments up to order~$t$.
    Specifically, a step of the walk on the unitary or orthognoal group of dimension~$2^{\mathsf n}$
    is a random Pauli rotation~$e^{\mathrm i \theta P /2}$.
    The spectral gap of this random walk is shown to be $\Omega(1/t)$,
    which coincides with the best previously known bound for a random walk on the permutation group on~$\{0,1\}^{\mathsf n}$.
    This implies that the walk gives an $\varepsilon$-approximate unitary $t$-design 
    in depth $\mathcal O(\mathsf n t^2 + t \log 1/\varepsilon)d$ 
    where $d=\mathcal{O}(\log \mathsf n)$ is the circuit depth to implement~$e^{\mathrm i \theta P /2}$.
    Our simple proof uses quadratic Casimir operators of Lie algebras.
    }
\end{quote}

\vspace{-2mm}

\section{Introduction}

An approximate unitary $t$-design~\cite{dankert2009unitary2design,gross2007designs} 
is an ensemble of unitaries that behaves similarly to the Haar random ensemble up to $t$-th moments.
For $\nn$-qubit ($\CC^2$) systems, there have been constructions of approximate unitary $t$-designs 
with circuit size $\poly(\nn,t)$~\cite{Brand_o_2016,Haferkamp2022randomquantum},
which have found wide applications in quantum information theory.
However, existing constructions using local random quantum circuits had rather steep dependence on~$t$.
In this paper, we consider random Pauli rotations,
which are $\exp(\ii \theta P/2)$ where $\theta$ is a random angle 
and $P$ is a random $\nn$-qubit Pauli operator.
We show that the product of $k$ independent random Pauli rotations
$e^{\ii \theta_k P_k/2} \cdots e^{\ii \theta_2 P_2/2} e^{\ii \theta_1 P_1/2}$
converges to a unitary $t$-design as $k$ increases.
\begin{theorem}\label{thm:main1}
    For any integers $\nn,t\geq 1$, it holds that
    \begin{equation}
        \norm*{\E_{\theta\sim(-\pi,\pi)}\E_{P \sim\pauli_\nn } \left(e^{\ii\frac\theta 2 P} \otimes e^{-\ii\frac\theta 2 \bar P}\right)^{\otimes t} - \E_{U \sim \SU(2^\nn)} (U \otimes \bar U)^{\otimes t}}\leq 1-\frac 1 {4t} - \frac 1 {4^\nn - 1}.
        \label{eq:mainbound}
    \end{equation}
    Here, $\pauli_\nn = \{ \one_2, \sigma^x, \sigma^y, \sigma^z\}^{\otimes \nn} \setminus \{ \one_{2^\nn} \}$ 
    is the set of all nonidentity $\nn$-qubit Pauli operators,
    the norm denotes the greatest singular value,
    $\bar U$ denotes the complex conjugate of $U$,
    the distributions for $P$, $\theta$, and $U$ are uniform in the designated domain. 
    
    In addition, for any finite dimensional unitary representation~$\rho$ of~$\SU(2^\nn)$, we have
    \begin{equation}
        \norm*{
            \E_{\theta\sim(-2\pi,2\pi)}\E_{P \sim\pauli_\nn } \rho(e^{\ii\frac\theta 2 P}) - \E_{U \sim \SU(2^\nn)} \rho(U)
        }
        \leq 1 - \frac 1 {4^\nn - 1}.\label{eq:univeral_spectral_gap}
    \end{equation}
\end{theorem}

\begin{corollary}\label{thm:main2}
    Consider two mixed unitary channels
    \begin{equation}
        \calC_t : \eta \mapsto \E_{P\sim \pauli_\nn, \, \theta \sim (-\pi,\pi)} \left(e^{\ii \frac \theta  2 P}\right)^{\otimes t} \eta \left(e^{-\ii \frac \theta  2 P}\right)^{\otimes t}
        \qquad\text{and}\qquad
        \calH_t : \eta \mapsto \E_{U \sim \SU(2^\nn)} U^{\otimes t} \eta U^{\dag \otimes t} 
    \end{equation}
    using the same distribution of~$P,\theta$ and~$U$ as in~\cref{eq:mainbound}. Then,
    \begin{equation}
    \begin{aligned}
        \norm*{\calC_t^k - \calH_t}_\diamond \le \eps \qquad &\text{ if } \qquad k \geq (4 \log 2) \nn t^2 + 4t \log \frac 1 \eps,\\
        (1-\eps)\calH_t \preceq  \calC^k_t \preceq(1+\eps)\calH_t\qquad &\text{ if } \qquad k \geq (4 \log 8) \nn t^2 + 4t \log \frac 1 \eps.
    \end{aligned}
    \end{equation}
    Here, $\norm*{\cdot}_\diamond$ denotes the diamond norm (completely bounded norm).
    Every instance $e^{\ii \theta_k P_k/2} \cdots e^{\ii \theta_1 P_1/2}$
    can be implemented using $\bigO{k \nn}$ $1$-qubit and any-to-any $\mathrm{CNOT}$ gates in depth~$\bigO{k\log \nn}$. 
\end{corollary}

We also give similar results for the special orthogonal groups in \cref{sec:orthogonaldesigns}.

\subsection{Previous spectral gap bounds}

Unless otherwise noted, $\NN$ stands for~$2^\nn$.

For a distribution~$\nu$ over~$\SU(\NN)$, 
the spectral gap~$\Delta(\nu,t)$ at $t$-th order\footnote{%
    In~\cite{bourgain2012spectral} the spectral gap
    means~$1 - \sup_\rho \norm*{\E_{U \sim \nu} \rho(U) - \E_{U \sim \SU(\NN)} \rho(U)}
    =1 - \sup_{\rho} \norm*{\E_{U \sim \nu} \rho(U)}$
    where the supremum is over all nontrivial finite dimensional unitary irreps of~$\SU(\NN)$.
    This is an immediate consequence of~\cref{thm:HaarProjectionTrivialRep} and the Peter--Weyl theorem.
    See also~\cite[Thm~3.9]{odonnell2023explicit}.
    Not all possible irreps of~$\SU(\NN)$ 
    appear in~$U \mapsto \bigoplus_{t \ge 1} (U \otimes \bar U)^{\otimes t}$.
}
is given by
\begin{equation}
    1 - \Delta(\nu,t) = \norm*{\E_{U \sim \nu} (U \otimes \bar U)^{\otimes t} - \E_{U \sim \SU(\NN)} (U \otimes \bar U)^{\otimes t}}.
\end{equation}
Consider the distribution of the product of~$k$ independent draws from~$\nu$, 
which corresponds to the $k$-fold convolution~$\nu^{*k}$.
Then, since $\calH_t = \E_{U \sim \SU(\NN)} (U \otimes \bar U)^{\otimes t}$ is a projector
and $\E_{U \sim \nu} (U \otimes \bar U)^{\otimes t}$ contains the image of~$\calH_t$
in the eigenspace of eigenvalue~$+1$ (see~\cref{thm:HaarProjectionTrivialRep}),
the spectral gap amplifies as~$1- \Delta(\nu^{*k},t) = (1-\Delta(\nu,t))^k$.
This allows us to exponentially improve the accuracy 
at the cost of linear blow-up in circuit size.
More generally, proving lower bounds on the spectral gap of the $t$-th moment operator 
is a standard approach to show that a random walk on a group
quickly converges to a $t$-wise independent distribution (often referred to as ``designs'').

Hence, a primary goal in efficient approximate unitary designs 
is to find~$\nu$ with $\poly(\nn)$ circuit size
with a large spectral gap, for example, $1/\poly(\nn,t)$.
A simple brickwall ``spacetime'' geometry of random unitary circuit has been shown to achieve this goal~\cite{Brand_o_2016},
whose analysis was recently improved~\cite{Haferkamp2022randomquantum}. 
Once the operator norm distance is bounded, 
one can convert it to additive or relative diamond distance.

As far as we know, the best previous spectral gap 
for any efficient approximate unitary design on an $\nn$-qubit system was
$\Omega(t^{-4-o(1)})$~\cite{Haferkamp2022randomquantum}. 
This work takes the ensemble of~\cite{Brand_o_2016},
where the circuit geometry is brickwall that uses local gates in a one-dimensional lattice.
Our ensemble does not have any geometric locality.
Note that~\cref{thm:main1} gives a lower bound~$4^{-\nn}$ on the spectral gap independent of~$t$.
Such a $t$-independent bound was also given in~\cite[Theorem~1]{Haferkamp2022randomquantum},
which reads $\Omega(\nn^{-5}4^{-\nn})$.

Similarly to unitary designs, the best previous spectral gap lower bound for 
the special orthogonal group~$\SO(\NN)$ had a large inverse-polynomial dependence on~$t$~\cite{odonnell2023explicit}, 
while the best previous spectral gap for the symmetric group~$S_\NN$ 
was $\Omega(t^{-1})$~\cite{Brodsky2008simple}.\footnote{%
    Here we ignore polynomial dependence in $\nn$ as it can be eliminated by taking powers.
}
Our spectral gap bounds for the special unitary and orthogonal groups
are thus the strongest in terms of~$t$ dependence,
and they coincides with the best known spectral gap for the symmetric group.

Some results on unitary design bypasses spectral gap analysis.
Aiming to minimize non-Clifford resources, 
\cite{Haferkamp_2022} analyzed 
alternating ``Clifford+$K$'' circuits and
the diamond distance of the associated mixed unitary channel 
to the Haar random channel directly.
Compared with the brickwall circuits,
our construction is conceptually closer to~\cite{Haferkamp_2022}.
However, their result is only applicable in the regime when $t=\bigO{\sqrt{\nn}}$.

\subsection{Implications}

\paragraph{Circuit complexity.}
By known reductions~\cite{Brandao2021models}, our result directly implies a lower bound for robust quantum circuit complexity 
of a product of $k$ random Pauli rotations. 
Specifically, let $U$ be a product of $k\ll 2^\nn$ random Pauli rotations,
which can be implemented by $\mathcal{O}(k\nn)$ gates. 
Then with high probability over the choice of $U$, 
any unitary $V$ satisfying $\norm{U-V}\leq 0.01$ 
must have quantum circuit complexity (the minimum number of gates to implement $V$) 
$\tilde{\Omega}(\sqrt{k\nn})$. 
Note that a robust square root circuit complexity lower bound 
was also established in~\cite{haferkamp2023moments};
however, the family of quantum circuits considered there used a non-universal gate set,
and therefore does not form an approximate unitary design.
A major open question is whether it is possible to construct distributions on~$\SU(2^\nn)$
using $\poly(\nn)$ size quantum circuits, 
such that the spectral gap is at least $1/\poly(\nn)$ and independent of $t$.
Such a result would imply a robust linear growth of quantum circuit complexity.

\paragraph{Seed length.}
Our unitary design requires sampling from a continuous interval~$(-\pi,\pi)$;
however, for given $t$, we can instead sample uniformly from a discrete set $\{m\pi/t:m\in\mathbb{Z}\cap [-t,t-1]\}$ (see~\cref{sec:discreteangles}).
Therefore, our distribution for $\varepsilon$-approximate unitary $t$-design 
is samplable using only $\bigO{t(\nn t+\log 1/\eps) (\nn +\log t)}$ random bits. 
Furthermore, instead of sampling each random Pauli rotation independently and uniformly at random, 
we can sample them in a pseudorandom way using a technique of~\cite{odonnell2023explicit}
which is a generalization of the derandomized graph squaring~\cite{Rozenman2005derandomized}. 
We can thus reduce the seed length to only $\mathcal{O}(\nn t+\log1/\varepsilon)$ 
by applying~\cite[Theorem 6.21]{odonnell2023explicit}. 
While this has the same scaling as the main result of~\cite{odonnell2023explicit}, 
our construction has the advantage of having explicit constants, 
as we do not rely on the implicit spectral gap of~\cite{bourgain2012spectral}.

\paragraph{Orthogonal designs and more.}

Our approach to unitary designs can be adapted to the special orthogonal groups $\SO(\NN)$ with parallel arguments.
The results are found in~\cref{sec:orthogonaldesigns}.
The seed length can be similarly reduced to $\mathcal{O}(\nn t+\log1/\varepsilon)$ with explicit constants.
This has been used to construct pseudorandom generators for halfspaces~\cite{odonnell2023explicit}.
Finally, we discuss quantum state designs in~\cref{sec:statedesign}, 
where we obtain better bounds than what would be obtained by directly applying our unitary design.

The analysis of the orthogonal groups is so similar to that of the unitary groups that
one might desire to have unified statements for all simple finite dimensional Lie groups.
However, as the representation theory of Lie groups is tackled in a case-by-case fashion in detail,
we find it best to analyze them separately.
Beyond the unitary and orthogonal groups, 
there is a family of symplectic groups,
which might have applications in classical Hamiltonian dynamics and quantum optics
as one often encounters symplectic spaces in these subjects.

\subsection{Overview of the argument}

We start by rewriting the tensor product in a different form, 
$\left(e^{\ii\frac\theta 2 P} \right)^{\otimes t} \otimes \left( e^{-\ii\frac\theta 2 \bar P} \right)^{\otimes t}
=
e^{\ii \theta \tau_*(P/2)}$, 
where
\begin{equation}
    \tau_*(P/2) = 
    \frac 1 2 \sum_{j=1}^t 
    \left(\one_\NN^{\otimes (j-1)} \otimes P\otimes \one_\NN^{\otimes (2t-j)}-\one_\NN^{\otimes (t+j-1)} \otimes \bar P\otimes \one_\NN^{\otimes (t-j)}\right).
\end{equation}
Note that for every $P\in\pauli_\nn$, the eigenvalues of $\tau_*(P/2)$ are exactly the integers in $[-t,t]$.
Thus, the averaging over $\theta$ gives $\E_{\theta\sim(-\pi,\pi)} e^{\ii \theta \tau_*(P/2)} = K_P$, 
where $K_P$ denotes the orthogonal projector onto the kernel of~$\tau_*(P/2)$.
Our goal is now reduced to analyzing the spectrum of $\E_{P\sim\pauli_\nn}K_P$. 
We calculate the norm of this exactly for the special case of~$\nn=1$ in~\cref{sec:su2}.
In general cases, we first block-diagonalize $K_P$ 
using the observation that $P \mapsto \tau_*(P)$ is a Lie algebra representation.
In each irrep, we upper bound the kernel projector by a quadratic approximation:
\begin{equation}\label{eq:opineq}
    K(H) \preceq \one - H^2 / \norm{H}^2
\end{equation}
where $K(H)$ is the kernel projector for a Hermitian operator~$H$,
which holds for any nonzero~$H$.
This inequality is useful because the kernel projector sum is then bounded by
a sum of squares of represented operators.
A nice property of Pauli operators is that this sum of squares of represented operators
is a scalar multiple of the identity for any irrep.
We then bound the scalar by $t,\NN$.

We use the representation theory of Lie algebras, 
but our exposition is elementary for the core bound in~\cref{thm:main1};
we assume no prior knowledge beyond the representation theory of~$\sutwo$ for the main bound.

\paragraph{Note added.}{We recently became aware of independent related work of A.~Bouland, C.~Chen, J.~Docter, P.~Hayden, and M.~Xu,
achieving similar results via a different construction~\cite{BCDHX}.}

\section{Lie algebras and probability distributions}

We begin with an observation that 
any unitary design can be regarded 
as a distribution on the linear space of a Lie algebra.
This will allow us to analyze spectral properties of a unitary design
by looking at certain hermitian operators in irreducible representations of~$\su(\NN = 2^\nn)$.
We will find the latter more convenient 
since our unitary design will have the most succinct description
as a distribution on the Lie algebra, rather than on the Lie group.

Often a Lie algebra is described by very concrete data,
called structure constants, $f^a_{bc}$, that enter in
the commutation relations as~$[J_b, J_c ] = \ii \sum_a f^a_{bc} J_a$
where $J_a$ are said to span the Lie algebra.
While this is mostly correct
and causes no trouble in practice,
the appearance of the imaginary unit~$\ii$ might bring some confusion.
So, we would like to clarify the complex and real coefficients.
The tangent space at the origin of a Lie group, taken as a real manifold, is a real Lie algebra~$\su(\NN;\RR)$.
This is a linear space over real numbers of traceless \emph{anti}hermitian matrices where a Lie bracket is defined by matrix commutator;
after all, the commutator of two hermitian operators is antihermitian, 
which is \emph{not} in the $\RR$-linear space of hermitian operators.
However, since a representation space is taken to be a complex vector space,
there is no reason not to allow complex coefficients in the span of antihermitian operators.
This extension of the coefficient field is formally called the complexification of the Lie algebra:
$\CC \otimes_\RR \su(\NN; \RR)$.
This complexified space consists of all $\CC$-linear combinations of traceless antihermitian operators,
which is the $\CC$-linear space of \emph{all} traceless matrices.
Hence, the complexification is perhaps better denoted 
as~$\mathfrak{sl}(\NN;\CC) = \CC \otimes_\RR \su(\NN;\RR)$,
the Lie algebra of special linear group.
In this paper, we take a liberal convention that 
\begin{itemize}
    \item when $\su(\NN)$ appears in the context of representation, we mean its complexificiation~$\mathfrak{sl}(\NN;\CC)$, and 
    \item when we discuss a probability distribution on~$\su(\NN)$, we mean a distribution on the real vector space of hermitian, rather than antihermitian, operators, with insertion of the imaginary unit~$\ii$, whenever needed, understood.
\end{itemize}

Suppose we have an $M$-dimensional representation~$\rho: \SU(\NN) \to \mathsf{U}(M)$ of~$\SU(\NN)$ for some~$M \ge 1$,
which may be reducible.
The representation map~$\rho$ is a Lie group homomorphism, 
and we have a corresponding commutative diagram~\cite[\S8.3]{FultonHarris} by the exponential map:
\begin{equation}
    \xymatrix{
        \su(\NN) \ar[r]^{\rho_*} \ar[d]^\exp & \mathfrak{u}(M) \ar[d]^\exp \\
        \SU(\NN) \ar[r]^\rho & \mathsf{U}(M)
    }\label{eq:LieAlgGroup}
\end{equation}
where $\rho_*$ is the induced, natural, Lie algebra homomorphism (a representation).
In the context of unitary designs, we are interested in the tensor representation~$\tau: \SU(\NN) \ni U \mapsto (U \otimes \bar U)^{\otimes t}$ so $M = \NN^{2t}$, whose induced Lie algebra representation $\tau_*$ is given by
\begin{equation}
    \tau_*\left( \frac P 2 \right) = \frac 1 2 \sum_{j=1}^t 
    (\one_\NN \otimes \one_\NN)^{\otimes (j-1)} 
    \otimes 
    ( P \otimes \one_\NN - \one_\NN \otimes \bar  P)
    \otimes (\one_\NN \otimes \one_\NN)^{\otimes(t-j)} 
    \label{eq:LieAlgRhoStar}
\end{equation}
for all traceless $\NN \times \NN$ matrix~$P$.

We define
\begin{equation}
    \pauli_{\nn} = \{\one_2, \sigma^x, \sigma^y, \sigma^z\}^{\otimes \nn} \setminus \{\one_{2^\nn}\},
\end{equation}
the set of all nonidentity tensor products of Pauli matrices
\begin{equation}
    \sigma^x = \begin{pmatrix} 0 & 1 \\ 1 & 0 \end{pmatrix},\quad
    \sigma^y = \begin{pmatrix} 0 & -\ii \\ \ii & 0 \end{pmatrix},\quad
    \sigma^z = \begin{pmatrix} 1 & 0 \\ 0 & -1 \end{pmatrix}.
\end{equation}
There are $\NN^2 -1 = 4^\nn - 1$ elements,
each of which is called a Pauli operator.
The factor of half in~\cref{eq:LieAlgRhoStar} is immaterial in that equation since $\tau_*$ is $\CC$-linear,
but we will keep it because

\begin{lemma}\label{thm:EigenvaluesOfRepresentedPaulisAreIntegers}
For any Pauli operator~$P \in \pauli_\nn$,
all the eigenvalues of~$\tau_*(P/2)$ are integers in $[-t,t]$,
and any integer in that range appears as an eigenvalue of~$\tau_*(P/2)$.
\end{lemma}

\begin{proof}
    All the summands of~\cref{eq:LieAlgRhoStar} are commuting with each other,
    so they are simultaneously diagonalizable, which amounts to setting $P = \sigma^z \otimes \one_{2^{\nn -1}}$,
    that is a diagonal matrix with $\pm 1$ on the diagonal. The lemma follows.
    Alternatively, we can think of~$\tau_*$ as 
    $
        \tau_* = (\one \oplus \ad)^{\otimes t} = \bigoplus_{k=0}^t \binom{t}{k} \ad^{\otimes k}
    $
    where $\one$ denotes the trivial representation and $\ad$ is the adjoint representation. 
    This is because the tensor product of the defining representation of~$\su(\NN)$ and its dual 
    is a direct sum of the trivial and the adjoint, each of which is irreducible.%
    \footnote{%
        On the trivial representation the representation map is zero,
        and on the adjoint we have~$\ad(P/2)  = \half (P \otimes \one - \one \otimes \bar P) |_{\su(\NN)}$ 
        where the restriction is on the linear space of all (vectorized) traceless $\NN$-by-$\NN$ matrices.
        Almost always, the adjoint representation map is explained by~$\ad(P/2)(X) = \half [P,X]$ for all~$X \in \su(\NN)$.
    }
    The represented operator~$\ad(P/2)$ has eigenvalues~$\pm 1$,
    and therefore $\ad^{\otimes t}$ has integer eigenvalues in~$[-t,t]$.
\end{proof}

We note that all Pauli operators~$P$ are equivalent to one another in any representation:

\begin{lemma}\label{thm:EqualSpectrum}
    For any representation~$\rho_*$ of~$\su(\NN)$ that is possibly reducible,
    and any nonidentity hermitian Pauli operator~$P$,
    the eigenvalue spectrum of the represented operator~$\rho_*(P)$ is independent of~$P$.
    In particular,
    the eigenvalue spectrum of~$\rho_*(P)$ is inversion symmetric about the origin; 
    that is, $\rho_*(P)$ and $-\rho_*(P)$ have the same spectrum.
\end{lemma}

\begin{proof}
    Any two nonidentity Pauli operator~$P,Q$ on $\nn$ qubits are congruent: 
    $P = U Q U^\dag$ for some~$U \in \SU(\NN)$.
    Exponentiating with $\theta \in \RR$ we have $e^{\ii \theta P} = U e^{\ii \theta Q} U^\dag$,
    and thus $\rho(e^{\ii \theta P} ) = \rho(U) \rho(e^{\ii \theta Q}) \rho(U)^\dag$.
    By~\cref{eq:LieAlgGroup} this translates to $e^{\ii \theta \rho_*(P)} = \rho(U) e^{\ii \theta \rho_*(Q)} \rho(U)^\dag$.
    Differentiating with respect to~$\theta$, we finally have $\rho_*(P) = \rho(U) \rho_*(Q) \rho(U)^\dag$.
    The last claim is because $P$ and $-P$ are congruent by some anticommuting Pauli operator.
\end{proof}

Now, we can consider probability distributions on~$\su(\NN)$ and their induced distributions on~$\SU(\NN)$.
For example, to assess a unitary design we have to analyze the distribution on~$\mathsf{U}(\NN^{2t})$ 
for various values of~$t$ induced by the tensor representation~$\tau$.
For a probability distribution~$\mu$ on the top left of the diagram~\cref{eq:LieAlgGroup},
we have corresponding distributions on all three other entries.
For any distribution~$\mu$ on~$\su(\NN)$, 
we denote an average with respect to~$\mu$ by~$\int_{\su(\NN)} \cdots \mu(X)\rd X$
where $X$ denotes any hermitian operator.\footnote{Another prevalent notation is $\int \cdots \rd \mu(X)$.}
In other words, $X \mapsto \mu(X)$ is the probability density ``function.''
For any distribution~$\mu$ on~$\su(\NN)$ and any integer~$t \ge 1$
we consider a linear operator on~$(\CC^\NN)^{\otimes 2t}$
\begin{equation}
    \calC_{\mu,t} = \int_{\su(\NN)} \exp(\ii X)^{\otimes t} \otimes \exp(-\ii \bar X)^{\otimes t} \mu(X) \rd X. \label{eq:calC}
\end{equation}
An obvious lemma will be useful:
\begin{lemma}\label{thm:subrep}
    If $\phi: \SU(\NN) \to \Aut(V)$ for $V \subseteq (\CC^\NN)^{\otimes 2t}$ is a subrepresentation 
    of~$\tau: U \mapsto (U \otimes \bar U)^{\otimes t}$,
    then
    \begin{equation}
        \calC_{\mu,t}|_V = \int_{\su(\NN)} \exp(\ii \phi_*(X)) \mu(X) \rd X. 
    \end{equation}
\end{lemma}

\begin{proof}
    $\phi$ is a Lie group representation, so the claim follows from the commutative diagram~\cref{eq:LieAlgGroup}.
    The assumption that $\phi$ is a subrepresentation of~$\rho$ is irrelevant.
\end{proof}

The Haar probability distribution on~$\SU(\NN)$, which we denote as~$\rd U$,
does give a distribution on~$\su(\NN)$ using the fact that the exponential map is one-to-one
on the open ball of radius~$\pi$ at the origin in the Schatten $\infty$-norm and is almost onto from that restricted domain,
but this is not very enlightening.
However, relevant averages can be succinctly described in terms of subrepresentations.

\begin{proposition}\label{thm:HaarProjectionTrivialRep}
    For any finite dimensional unitary representation~$\rho$ of a compact Lie group~$G$, 
    the integral~$\int_{G} \rho(U) \rd U$ with respect to the Haar measure
    is the orthogonal projector onto the trivial subrepresentation subspace of~$\rho$.
    In particular, for any integer $t \ge 1$ the Haar average
    \begin{equation}
    \calH_t = \int_{\SU(\NN)} (U \otimes \bar U)^{\otimes t} \rd U   
    \end{equation}
    is the orthogonal projector
    onto the trivial subrepresentation subspace of $\tau: U \mapsto (U \otimes \bar U)^{\otimes t}$
    within~$(\CC^{\NN})^{\otimes 2t}$.
\end{proposition}

Note that in~\cref{thm:main1} we denoted the Haar random mixed unitary channel by~$\calH_t$,
but here we overload the notation to mean its linearized map.
The representation~$\rho$ does not have to be finite dimensional,
but we do not discuss any infinite dimensional representations in this paper.

\begin{proof}
    Let $\calH=\int_G \rho(U) \rd U$.
    Since the Haar measure is left invariant, 
    the represented unitary~$\rho(V)$ for any~$V \in G$ acts by the identity on the image of~$\calH$.
    It follows that~$\calH^2 = \calH$.
    Since $\rho$ is a unitary representation, $\calH^\dag = \int_G \rho(U^{-1}) \rd U$.
    Since $U \mapsto U^{-1}$ is a measure-preserving homeomorphism of~$G$ onto itself, 
    $\calH^\dag = \calH$. 
    Let $\mathcal V$ be the representation space.
    The trivial representation subspace is $\mathcal T = \{ v \in \mathcal V : \rho(U) v = v, \ \forall U \in G\}$.
    If $v = \calH w$ for some $w \in \mathcal V$,
    then $\rho(U) v = \rho(U) \calH w = \calH w = v$, and hence $v \in \mathcal T$.
    So, the image of $\calH$ is contained in the trivial representation.
    If $v \in \mathcal T$, then $\calH v = v$, showing that $v$ is in the image of~$\calH$.
\end{proof}

\section{Random Pauli rotations}

Now we consider more concrete distributions on~$\su(\NN)$ where $\NN = 2^\nn$ for an integer~$\nn \ge 1$.

\begin{definition}
For any~$P \in \pauli_\nn$
we define a distribution, called a \emph{random Pauli rotation by~$P$},
as the uniform probability measure on~$\{ \ii \theta P /2  \in \su(\NN;\RR) \,|\, \theta \in (-\pi,\pi) \subset \RR \}$.
A \emph{random Pauli rotation} with respect to a discrete probability distribution~$\{(P,\Pr[P]) | P \in \pauli_\nn \}$ 
on~$\pauli_\nn$
is the probabilistic mixture~$\sum_P \Pr[P] \mu_P$ of random Pauli rotations~$\mu_P$ by~$P$. 
\end{definition}

We will only use the uniform distribution over~$\pauli_\nn$,
but it may be helpful to proceed with a general distribution on~$\pauli_\nn$.

\begin{lemma}\label{thm:ThetaAverage}
For a random Pauli rotation~$\mu = \sum_P \Pr[P] \mu_P$,
the average operator~$\calC_{\mu,t}$ in~\cref{eq:calC} restricted to a subrepresentation~$\phi_* : \su(\NN) \to \Aut(V)$ of~$\su(\NN)$ 
within the tensor representation~$\tau : U \mapsto (U \otimes \bar U)^{\otimes t}$,
simplifies as
\begin{equation}
    \calC_{\mu,t}|_V = \int_{\su(\NN)} \exp(\ii \phi_*(X)) \mu(X) \rd X 
    = \sum_P \Pr[P] K(\phi_*(P/2))
\end{equation}
where $K(H)$ for any hermitian operator~$H$ is the orthogonal projector onto~$\ker H$, the eigenspace of eigenvalue zero.
\end{lemma}

\begin{proof}
The first equality is noted in~\cref{thm:subrep}.
For the second equality,
it suffices to evaluate $\calC_{\mu_P,t}|_V$ for a random Pauli rotation~$\mu_P$ by~$P$.
We have observed in~\cref{thm:EigenvaluesOfRepresentedPaulisAreIntegers}
that all the eigenvalues of~$\tau_*(P/2)$ are integers.
A subrepresentation of~$\tau_*$ is nothing but a block-diagonal piece of~$\tau_*$ 
after a unitary basis change on~$(\CC^\NN)^{\otimes 2t}$,
so the eigenvalues of~$\phi_*(P/2)$ can only be a subset of those of~$\tau_*(P/2)$.
Hence, the average of~$e^{\ii \theta \phi_*(P/2)}$ over~$\theta$ eliminates all eigenspaces of nonzero eigenvalues.
\end{proof}

\begin{remark}
    A Pauli operator~$P$, a tensor product of hermitian Pauli matrices, is a traceless unitary of eigenvalues~$\pm 1$.
    Observe that $\ii P$ is a member of~$\SU(\NN)$ and also of~$\su(\NN;\RR)$ where~$\NN$ is a power of~$2$.
    For example, $\ii \sigma^z \in \SU(2) \cap \su(2;\RR)$
    and $\ii \sigma^z \otimes \sigma^z \in \SU(4) \cap \su(4;\RR)$.
    If $\rho$ is a Lie group representation map, and $\rho_*$ is the derived Lie algebra representation map,
    then we may consider $\rho(\ii P)$ and $\rho_*(\ii P)$
    both of which are some matrices of the same dimension.
    Generally, $\rho(\ii P) \neq \rho_*(\ii P)$.
    However, since $P^2 = \one$, we instead have $\exp(\ii \pi P / 2) = \cos(\pi/2)\one + \ii \sin(\pi/2) P = \ii P$
    and therefore     by~\cref{eq:LieAlgGroup} we have
    \begin{align}
        \exp(\ii \pi \rho_*(P)/2) = \exp(\rho_*(\ii \pi P/2)) = \rho (\exp(\ii\pi P/2)) = \rho(\ii P).
    \end{align}
\end{remark}

\begin{proposition}\label{thm:KernelSum}
    For any random Pauli rotation with respect to~$\{ (P, \Pr[P]) \}$ at order~$t$, we have
    \begin{equation}
        \norm{\calC_{\mu,t} - \calH_t} = \max_{\phi}\norm*{\sum_P \Pr[P] K(\phi_*(P/2))}
    \end{equation}
    where $\phi$ ranges over all irreducible nontrivial subrepresentations of
    the tensor representation~$\tau: U \mapsto (U \otimes \bar U)^{\otimes t}$.
\end{proposition}

\begin{proof}
Immediate from~\cref{thm:ThetaAverage} and~\cref{thm:HaarProjectionTrivialRep}.
\end{proof}

It is known~\cite[Lemma~3.7]{Harrow2009random} that the spectral gap, $1-\norm{\calC_{\mu,t} - \calH_t}$,
is nonzero positive
if~$\{ (P, \Pr[P]) \}$ induces a ``densely generating'' distribution on~$\SU(\NN)$.

The motivation for us to consider the quantum circuit of random Pauli rotations 
is its simple implementation:
\begin{proposition}
    Suppose that for an $\nn$-qubit system,
    $\mathrm{CNOT}$ can be applied only across a set of unordered pairs of qubits.
    This defines an undirected simple graph (``connectivity graph'') over qubits, which we assume is connected.
    For any $P\in \pauli_{\nn}$ and $\theta \in \RR$, a unitary $e^{\ii \frac{\theta}{2} P}$ 
    can be implemented using (1) one $1$-qubit Pauli $X$ rotation $e^{\ii \frac{\theta}{2} \sigma^x}$, 
    (2) at most $2\nn$ $1$-qubit Hadamard and Phase gates, 
    and (3) at most $2\nn-2$ $\mathrm{CNOT}$ and $\mathrm{SWAP}$ gates.
\end{proposition}

\begin{proof}
    It suffices to find a sequence of gates that maps $e^{\ii \frac{\theta}{2} P}$ to $e^{\ii \frac{\theta}{2} \sigma^x}\otimes \one_2^{\otimes(\nn-1)}$ by conjugation. 
    We first apply Hadamard and Phase gates by conjugation to obtain $e^{\ii \frac{\theta}{2} Q}$ where
    $Q = Q_1\otimes \cdots \otimes Q_\nn$ is a tensor product of~$\sigma^x$'s and~$\one_2$'s.
    In the connectivity graph, we assign each node a binary value corresponding to the support of $Q$,
    {\em i.e.}, $v_i = 1$ if and only if $Q_i=\sigma^x$. 
    Every connected graph has a spanning tree. 
    For each edge $(v_\mathrm{parent}, v_\mathrm{child})$ in the spanning tree such that all the children of $v_\mathrm{child}$ are zeros, if $v_\mathrm{parent}=v_\mathrm{child}=1$, apply a CNOT gate by conjugation on the corresponding two qubits. If $v_\mathrm{parent}=0$ and $v_\mathrm{child}=1$, apply a SWAP gate by conjugation on the corresponding two qubits. Both operations will result in $v_\mathrm{parent}=1$, $v_\mathrm{child}=0$, {\em i.e.}, all the children of $v_\mathrm{parent}$ are zeros. This procedure terminates when the only nonzero node is the root, which corresponds to $e^{\ii \frac{\theta}{2} \sigma^x}\otimes \one_2^{\otimes (n-1)}$. The total number of CNOT and SWAP gates applied is at most twice the number of edges in the spanning tree, which is $2(\nn-1)$.
\end{proof}

\begin{figure}[t]
    \begin{center}
    \begin{equation*}
        e^{\ii \frac{\theta}{2} P} \quad=\quad\begin{quantikz}[slice style=blue] 
        &\slice[style=gray]{}&\targ{} \slice[style=gray]{}&\slice[style=gray]{}&\slice[style=gray]{}&\slice[style=gray]{}&\slice[style=gray]{}&\slice[style=gray]{}&\targ{}\slice[style=gray]{}&&\\ 
        &\gate{H}&\ctrl{-1}&\targ&&&&&\targ{}&\ctrl{-1}&\gate{H}& \\
        &&\targ&&&&&&&\targ{}&&\\
        &\gate{H}&\ctrl{-1}&\ctrl{-2}&\ctrl{1}&\gate{e^{\mathrm{i}\frac{\theta}{2} X}}&\ctrl{1}&\ctrl{-2}&\ctrl{-1}&\gate{H}&\\
        &&\ctrl{1}&\ctrl{2}&\targ{}&&\targ{}&\ctrl{2}&\ctrl{1}&&\\
        &&\targ&&&&&&&\targ{}&&\\
        &&\ctrl{1}&\targ{}&&&&\targ{}&\ctrl{1}&&\\
        &\gate{H}&\targ&&&&&&&\targ{}&\gate{H}&
        \end{quantikz}
    \end{equation*}
        
    \end{center}
    
    \caption{Implementation of $e^{\ii \frac{\theta}{2} P}$ ($P\in \pauli_{\nn}$) by an $\bigO{\log \nn}$ depth circuit. The example corresponds to the Pauli string \texttt{XZXZXXXZ}. Gates between two dashed lines are implemented in parallel.}
    \label{fig:paulirotation}
\end{figure}
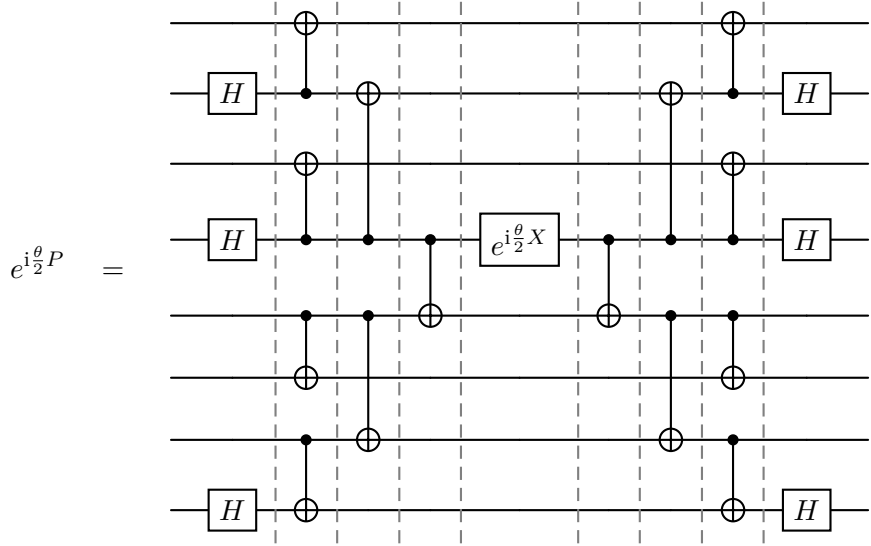

\begin{corollary}\label{thm:singlePauliRotationComplexity}
    With all-to-all connectivity, for any $P\in \pauli_{\nn}$ and $\theta \in \RR$, 
    the unitary~$e^{\ii \frac{\theta}{2} P}$ can be implemented using $\bigO{\nn}$ $1$-qubit and $\mathrm{CNOT}$ gates in circuit depth $\bigO{\log \nn}$. See \cref{fig:paulirotation} for an example.
\end{corollary}

\section{The special case of \texorpdfstring{$\sutwo$}{su(2)}}\label{sec:su2}

There are only three Pauli operators $\sigma^x, \sigma^y, \sigma^z$ (up to real scalars) in~$\sutwo$ 
so a random Pauli rotation is specified by $\Pr[\sigma^x], \Pr[\sigma^y], \Pr[\sigma^z]$.
The goal is clear in~\cref{thm:KernelSum}.
With an $\sutwo$-irrep~$\phi_*$ in mind, we just write $J_{x,y,z}$ to mean $\phi_*(\sigma^{x,y,z}/2)$.
Note the factor of~$2$ in the denominator, 
which gives $[J_a,J_b] = \ii J_c$ where $(a,b,c)$ is a cyclic permutation of~$(x,y,z)$.
We have to calculate the spectral norm of
\begin{equation}
    \Pr[\sigma^x] K(J_x) + \Pr[\sigma^y] K(J_y) + \Pr[\sigma^z] K(J_z)
\end{equation}
for all irreps that appear in the tensor representation~$\tau : \SU(2) \ni U \mapsto (U \otimes \bar U)^{\otimes t}$.
Since the dual of the defining irrep of~$\sutwo$ is equivalent to itself,
the representation~$\tau$ is simply the $2t$-fold tensor product of the defining irrep of~$\sutwo$.
It is a standard fact that every irrep that appears in~$\tau$ is odd dimensional 
because $2t$ is an even number,
and every irrep of dimension~$2\ell+1$ appears in~$\tau$ where~$\ell \le t$.

\begin{lemma}\label{thm:sutwoKOverlap}
    Each of~$K(J_x), K(J_y), K(J_z)$ has rank~$1$ on any nontrivial $\sutwo$-irrep of odd dimension~$2\ell+1$
    for an integer $\ell \ge 1$.
    There exist normalized vectors $\ket{x} \in \ker J_x, \ket{y} \in \ker J_y, \ket{z} \in \ker J_z$
    such that
    \begin{equation}
        \braket{x|y} = \braket{y|z} = \braket{z|x} = 
        \begin{cases}
            \frac{(-1)^{\ell/2}}{2^\ell}\binom{\ell}{\ell/2} & \text{if $\ell$ is even},\\
            0 &\text{otherwise}.
        \end{cases}
    \end{equation}
\end{lemma}

In view of~$\so(3) \cong \su(2)$, an odd dimensional irrep is often called an ``integer spin'' representation,
and an even dimensional irrep a ``half-integer spin'' representation.
For even dimensional irreps, it is well known that~$K(J_x) = K(J_y) = K(J_z) = 0$.

\begin{proof}
    Let $\phi_*$ be an $\sutwo$-irreducible representation map acting on~$V \cong \CC^{2\ell+1}$.
    It is a standard fact that $J_z$ has eigenvalues~$\ell,\ell-1,\ldots,-\ell+1, -\ell$,
    each with multiplicity~$1$. Hence, the kernel is one-dimensional.

    (\emph{First method by symmetric powers})
    Note that $U = \exp(-\ii \frac \pi {3\sqrt{3}} (\sigma^x + \sigma^y + \sigma^z)) \in \SU(2)$ 
    acts by conjugation as $\sigma^x \mapsto \sigma^y \mapsto \sigma^z \mapsto \sigma^x$.
    For a normalized vector~$\ket z \in \ker J_z \subset V$,
    the vector~$\ket x = \phi(U)\ket z$ spans~$\ker J_x$ 
    and $\ket y = \phi(U^2)\ket z$ spans~$\ker J_y$.
    So, the three inner products in the claim are the same.
    It remains to calculate~$\braket{z|x} = \bra z \phi(U) \ket z$.

    A concrete expression for~$\phi$ is obtained by considering 
    the $2\ell$-fold symmetric power of the defining representation of~$\sutwo$~\cite[(11.8)]{FultonHarris}. 
    Let $\ket 0$ and~$\ket 1$ be a basis of~$\CC^2$ such that 
    $\sigma^z \ket 0 = +\ket 0$ and $\sigma^z \ket 1 = - \ket 1$.
    Then,
    \begin{equation}
    U = \frac{e^{-\ii \pi/4}}{\sqrt 2} \begin{pmatrix} 1 & -\ii \\ 1 & \ii \end{pmatrix}.
    \end{equation}
    Written in terms of vectors in~$(\CC^2)^{\otimes 2\ell}$,
    a set of basis vectors of~$V$ can be chosen to be
    \begin{equation} 
        \binom{2\ell}{k}^{-1/2} \sum_{w \in \{0,1\}^{2\ell}: \abs w = k} \ket w
    \end{equation}
    for $k = 0,1,2,\ldots,2\ell$, 
    where $\abs w$ is the number of $1$'s in the bitstring~$w$ of length~$2\ell$.
    These are eigenvectors of~$J_z$ with eigenvalues~$\ell-k$. 
    So,
    \begin{equation}
        \ket z = \binom{2\ell}{\ell}^{-1/2} \sum_{w: \abs w = \ell} \ket w.
    \end{equation}
    Applying $\phi(U) = U^{\otimes 2\ell} \big|_V$, we will obtain~$\ket x$.
    With~$\ket + = 2^{-1/2}(\ket 0 + \ket 1)$ and~$\ket - = 2^{-1/2}(\ket 0 - \ket 1)$,
    we have $U\ket 0 = e^{-\ii \pi/4} \ket +$ and $U \ket 1 = - e^{\ii \pi/4} \ket -$.
    Hence,
    \begin{equation}
        \braket{z|x} = \bra z U^{\otimes 2\ell} \ket z 
        =
        \sum_{w \in \{0,1\}^{2\ell}: \abs w = \ell} \bra w (-1)^\ell \ket{+}^{\otimes \ell}\!\ket{-}^{\otimes \ell}
    \end{equation}
    where the second equality is because
    both~$\ket z$ and~$\bra z$ are invariant under permutations of tensor factors.
    For a bitstring~$w$ of length~$2\ell$ with~$\abs w = \ell$, 
    let $m$ be the number of $1$'s in the last $\ell$ bits.
    Then, $\braket{w|+^\ell -^\ell} = 2^{-\ell} (-1)^m$. 
    There are $\binom{\ell}{\ell-m} \binom{\ell}{m}$ such bitstrings,
    and $m$ ranges from~$0$ to~$\ell$.
    So,
    \begin{equation}
        \braket{z|x} = \frac{(-1)^\ell}{2^\ell} \sum_{m=0}^\ell (-1)^m \binom{\ell}{\ell - m}\binom{\ell}{m} .
    \end{equation}
    The sum is the coefficient of~$h^{\ell}$ in 
    a polynomial~$(1+h)^{\ell}(1-h)^{\ell} = (1-h^2)^\ell$ in a variable~$h$.
    There is no $h^\ell$ term if~$\ell$ is odd, implying that the sum is zero.
    If $\ell$ is even, then the coefficient is~$(-1)^{\ell/2} \binom{\ell}{\ell/2}$.
    This completes the proof.
        
    (\emph{Second method using raising and lowering operators})
    Let $\ket \ell$ be an eigenstate of~$J_z$ with eigenvalue~$\ell$: $J_z \ket \ell = \ell \ket \ell$.
    Define $J^+ = J_x + \ii J_y$ and $J^- = J_x - \ii J_y$,
    and inductively $J^- \ket k = a_{-k} \ket{k-1}$ for $k=\ell, \ell-1,\ldots, -\ell +1$
    where
    \begin{equation}
        a_k = \sqrt{\ell(\ell+1) - k (k+1)} = a_{-k-1} .
    \end{equation}
    Here, $\ket 0 = \ket z$.
    It is straightforward to check that $J^+ \ket{k} = a_k \ket{k+1}$;
    the vector~$\ket{\ell+1}$ is never defined, but~$a_\ell = 0$.
    
    Since $J_y \ket y = 0$, we have $J^+ \ket y = J^- \ket y$.
    This implies that $\ket{\ell -1}$ cannot be a nonzero component of~$\ket y$,
    which implies, in turn, that $\ket{\ell - 2j -1}$ for any integer~$j$ cannot be a nonzero component of~$\ket y$.
    Hence, $\ket y$ is in the span of~$\ket{\ell}, \ket{\ell-2}, \ldots, \ket{-\ell +2}, \ket{-\ell}$.
    In particular, if $\ell$ is odd, $\braket{z|y} = \braket{0|y} = 0$.

    Suppose $\ell = 2p$, an even integer, and put $\ket y = \sum_{k=-p}^p c_k \ket{2k}$.
    Then, the equation $J^+ \ket y = J^- \ket y$ implies that
    \begin{equation}
        c_k a_{2k} = c_{k+1} a_{-2k-2} = c_{k+1} a_{2k+1} . \label{eq:c-and-a}
    \end{equation}
    One can verify that $c_k = c_{-k}$ and
    \begin{align}
        \frac{a_{2k}^2}{a_{2k+1}^2} &= \frac{2p(2p+1) - 2k(2k+1)}{2p(2p+1) - (2k+1)(2k+2)} = \frac{(p-k)(2p+2k+1)}{(p+k+1)(2p-2k-1)} \\
        &= \frac{(p-k)^2 (2p+2k+2)(2p+2k+1)}{(2p-2k)(2p-2k-1)(p+k+1)^2} 
        = 
        \frac{\binom{2p-2k-2}{p-k-1}\binom{2p+2k+2}{p+k+1}}{\binom{2p-2k}{p-k}\binom{2p+2k}{p+k}}. \nonumber
    \end{align}
    Therefore, 
    \begin{equation}
        \frac{\abs{c_{k}}^2}{\abs{c_0}^2} = \frac{\binom{2p-2k}{p-k}\binom{2p+2k}{p+k} }{\binom{2p}{p}^2}. 
    \end{equation}
    Since $\braket{y|y} = 1$, we must have 
    \begin{equation}
        1 = \sum_{k = -p}^p \abs{c_k}^2 = \frac{\abs{c_0}^2}{\binom{2p}{p}^2} \sum_{k = -p}^p \binom{2p-2k}{p-k}\binom{2p+2k}{p+k} = \frac{\abs{c_0}^2}{\binom{2p}{p}^2} 4^{2p},
    \end{equation}
    where the last equality follows from a combinatorial identity\footnote{This can be proved by, for example, a formula $(1-4h)^{-1/2} = \sum_{n=0}^\infty \binom{2n}{n} h^n$ and its square, where $h$ is a variable.}~\cite{Sved_1984}
    \begin{align}
        \sum_{i=0}^n \binom{2i}{i}\binom{2n-2i}{n-i} = 4^{n}. 
    \end{align}
    This shows that $\abs{\braket{y|z}} = \abs{c_0} = \frac{1}{4^p}\binom{2p}{p}$ if~$\ell = 2p$.    
    
    The complex phase~$\alpha$ of~$\braket{y|z} = \alpha \frac{1}{4^p}\binom{2p}{p}$ is not fixed by the normalization,
    but~$\braket{x|y}\!\braket{y|z}\!\braket{z|x}$ is well defined regardless of~$\alpha$.
    To evaluate this product of the three inner products, we may use any normalized vectors in the kernels.
    A vector $\ket x \in \ker J_x$ can be computed by solving~$J^+ \ket x = - J^- \ket x$.
    By completely parallel calculation, we find a solution
    $\ket x = \sum_{k=-p}^p (-1)^k c_k \ket{2k}$.
    Then,
    \begin{equation}
        \beta = \braket{x|y}\!\braket{y|z}\!\braket{z|x} = \braket{x|y}\!\braket{y|0}\!\braket{0|x} = \abs{c_0}^2 \sum_{k=-p}^p (-1)^k \abs{c_k}^2 .
    \end{equation}
    This is a real number, which means that we may take $\alpha = \pm 1 = \beta / \abs{\beta}$. 
    We know that $\abs{\beta} = \abs{c_0}^3 < 1$, so $\sum_{k=-p}^p (-1)^k \abs{c_k}^2 = \pm \abs{c_0}$.

    From~\cref{eq:c-and-a} we know that $\abs{c_0}<\abs{c_1}<\cdots<\abs{c_p}$. 
    Suppose $p$ is odd and $\beta > 0$.
    Then, we must have $\sum_{k=-p}^p (-1)^k \abs{c_k}^2 = \abs{c_0}$,
    which gives a contradiction:
    $0 = (- \abs{c_0} - \abs{c_1}^2) + (\abs{c_0}^2 - \abs{c_1}^2) + 2(\abs{c_2}^2 - \abs{c_3}^2) + \cdots + 2(\abs{c_{p-1}}^2 - \abs{c_p}^2) < 0$.
    Therefore, $\beta < 0 $ if $p$ is odd.
    Similarly, suppose $p$ is even and $\beta < 0$.
    Then, we must have $\sum_{k=-p}^p (-1)^k \abs{c_k}^2 = - \abs{c_0}$,
    which gives a contradiction:
    $ 0 = (\abs{c_0} + \abs{c_0}^2) + 2(- \abs{c_1}^2 + \abs{c_2}^2) + \cdots + 2(-\abs{c_{p-1}}^2 + \abs{c_p}^2) > 0$.
    Therefore, $\beta > 0$ if $p$ is even.
    This completes the proof.
\end{proof}

\begin{corollary}\label{thm:sutwoSpectralGap}
    For a random Pauli rotation~$\mu$ with respect to~$\{(\sigma^x,\frac 1 3),(\sigma^y,\frac 1 3),(\sigma^z,\frac 1 3)\}$
    we have 
    \begin{equation}
        \norm*{\calC_{\mu,t} - \calH_t} = \frac 1 {12} \cdot
        \begin{cases}
            4 & (t=1)\\
            6 & (t=2,3)\\
            7 & (t \ge 4)
        \end{cases} .
    \end{equation}
\end{corollary}

\begin{proof}
    We calculated $f^2 = \Tr(K(J_x)K(J_y))$, etc, in~\cref{thm:sutwoKOverlap},
    where $(-1)^{\ell/2} f = \binom{\ell}{\ell/2} / 2^\ell$ is a decreasing function in even integer~$\ell$.
    Hence by~\cref{thm:KernelSum}, there are only three cases to check: 
    $f=0$ if $\ell$ is odd, $f = -1/2$ if $\ell = 2$, and $f = 3/8$ if $\ell = 4$.

    Let $M = K(J_x) + K(J_y) + K(J_z)$ and $a,b,c \in [0,3] \subset \RR$ be the eigenvalues of~$M$.
    Then $\Tr(M) = a+b+c = 3$, $\Tr(M^2) = a^2 + b^2 + c^2 = 3 + 6 f^2$, 
    and $\Tr(M^3) = a^3 + b^3 + c^3 = 3 + 18 f^2 + 6 f^3$.
    Calculation gives $a = b = 1-f$ and $c = 1+2f$ up to permutations,
    so $\norm M = \max(1-f,1+2f)$.
    This norm reaches the maximum~$7/4$ when $\ell = 4$ and $f = 3/8$. 
\end{proof}

\section{A spectral gap bound by quadratic Casimir invariants}\label{sec:casimir}

\begin{theorem}\label{thm:SpectralGapByQuadraticCasimir}
    Let $\mu$ be the random Pauli rotation with respect to the uniform distribution on~$\pauli_\nn$.
    Then, with $\NN = 2^\nn$ and for any integer $t \ge 1$,
    \begin{equation}
        \norm{\calC_{\mu,t} - \calH_t} \le 1 - \frac{1}{4t} \frac{\NN^2}{\NN^2-1} - \frac{1}{\NN^2 - 1}. \label{eq:maingapbound}
    \end{equation}
    If $t \le \NN/2$, then
    \begin{equation}
        \norm{\calC_{\mu,t} - \calH_t} \le 1 - \frac{1}{2t} \frac{\NN(\NN-t+1)}{\NN^2 - 1}.\label{eq:smalltgapbound}
    \end{equation}
\end{theorem}

Comparing with~\cref{thm:sutwoSpectralGap}, we see that with $\NN=2$ 
the inequality in~\cref{eq:maingapbound} is saturated if and only if $t \in \{1,2,4\}$.

\begin{proof}
    It follows from~\cref{thm:KernelSum} that
    \begin{equation}
        \norm{\calC_{\mu,t} - \calH_t} = \max_J \norm*{\E_{P \in \pauli_\nn} K(J_P)}
    \end{equation}
    where the maximum is over all nontrivial $\su(\NN)$-irreps that appear 
    in the tensor representation~$\tau : U \mapsto (U \otimes \bar U)^{\otimes t}$.
    By~\cref{eq:LieAlgRhoStar}, the norm~$\ell$ of a represented operator~$J_P$ in a nontrivial $\su(\NN)$-irrep
    is a nonzero integer and is at most~$t$.
    We use an operator inequality
    \begin{equation}\label{eq:kernelupperbound}
        K(H) \preceq \one - H^2 / \norm{H}^2
    \end{equation}
    which holds for any nonzero hermitian operator~$H$,
    where $K(H)$ is the orthogonal projector onto~$\ker H$.
    Averaging over~$\pauli_\nn$, we have that
    \begin{equation}\label{eq:avgkernelprojbound}
        \E_{P \in \pauli_\nn} K(J_P) 
        \preceq 
        \one - \E_{P \in \pauli_\nn} \frac{J_P^2}{\norm{J_P}^2}. 
    \end{equation}
    It is well known (by a more general argument) that
    the last term is a scalar multiple of the identity, called a Casimir operator or invariant,
    where the scalar depends only on the irrep.
    We give an elementary calculation to this end in~\cref{thm:quadraticSumIsInvariant} below.
    By the lower bound in~\cref{thm:QuadraticCasimirBound} below,
    \begin{equation}
        \one - \E_{P \in \pauli_\nn} \frac{J_P^2}{\norm{J_P}^2}
        \preceq \one \parens*{
            1 - \frac{1}{4\ell} \frac{\NN^2}{\NN^2 - 1} - \frac{1}{\NN^2 -1}
        }
    \end{equation}
    where $\ell = \norm{J_P}$ is independent of~$P$.
    Since $\ell \le t$, we complete the proof of~\cref{eq:maingapbound}.
    For~\cref{eq:smalltgapbound}, we use~\cref{thm:improvedlowerbound} further below.
\end{proof}

\begin{lemma}\label{thm:quadraticSumIsInvariant}
For any $\su(\NN)$-irrep $\phi_*$, 
the quadratic sum $\sum_{P \in \pauli_\nn} \phi_*(P/2)^2$
is a scalar multiple of the identity.
\end{lemma}

\begin{proof}
    Abbreviate $\phi_*(P/2)$ by~$J_P$.
    We show by direct calculation that $\sum_P J_P^2$ 
    commutes with~$J_Q$ for all~$Q \in \pauli_\nn$.
    Since the commutator obeys the Leibniz rule, we have 
    $[J_Q, J_P^2] = J_P [J_Q, J_P] + [J_Q, J_P] J_P$.
    This may be nonzero only if $PQ = -QP = \pm \ii R$ for some~$R \in \pauli_\nn$.
    For an anticommuting pair~$P,Q$, the Pauli operator~$R$ also anticommutes with each of~$P,Q$.
    So, the subset of all Pauli operators that anticommute with~$Q$ is partitioned into unordered pairs $\{P, R\}$
    where $[J_P, J_Q] = \ii J_R$.
    That is, for each pair~$\{P,R\}$, the three elements~$J_P, J_Q, J_R$ span~$\sutwo$.
    Then,
    \begin{align}
        [J_Q, J_P^2 + J_R^2]
        &=
        J_P [J_Q, J_P] + [J_Q, J_P] J_P + J_R [J_Q, J_R] + [J_Q, J_R] J_R \nonumber\\
        &=
        J_P (-\ii J_R) + (-\ii J_R) J_P + J_R (\ii J_P) + (\ii J_P) J_R \label{eq:sutwoCasimir}\\
        &= 0. \nonumber
    \end{align}
    Since we are working with an irrep, Schur's lemma implies that $A = \sum_P J_P^2$ is proportional to the identity.
\end{proof}

\begin{lemma}\label{thm:QuadraticCasimirBound}
Let $J_P = \phi_*(P/2)$ be the represented operator in an $\su(\NN)$-irrep~$\phi_*$ for any Pauli operator~$P \in \pauli_\nn$
where $\NN = 2^\nn \ge 2$,
and let $\ell = \norm{J_P}$ be the Schatten $\infty$-norm, which is independent of~$P$.
Then,
\begin{equation}
    \parens*{\frac{\NN^2 \ell }{4}  + \ell^2}\one 
    \preceq 
    \sum_{P \in \pauli_\nn} J_P^2 
    \preceq 
    \parens*{ \frac{\NN (\NN-1) \ell }{2}  + (\NN-1)\ell^2 }\one .
\end{equation}
For any $\nn \in \ZZ_{>0}$ and $\ell \in \half \ZZ_{>0}$, there is an $\su(2^\nn)$-irrep that saturates the upper bound.
For any $\nn,k \in \ZZ_{>0}$, there is an $\su(2^\nn)$-irrep with $\ell = 2^{\nn-2} k$ that saturates the lower bound.
The saturating irreps are unique up to isomorphisms.
\end{lemma}

In this lemma it is not required that $\phi_*$ is a subrepresentation of 
a tensor representation~$\tau: U \mapsto ( U \otimes \bar U)^{\otimes t}$.

The lower bound proof can be understood using just the representation theory of~$\sutwo$.

\begin{proof}[Proof of the lower bound]
    The norm $\ell = \norm{J_P}$ is independent of~$P$ by~\cref{thm:EqualSpectrum}.
    \Cref{thm:quadraticSumIsInvariant} says that $A = \sum_P J_P^2$ is a scalar multiple of the identity.
    We have to estimate \emph{the} eigenvalue of~$A$.
    It suffices to examine the action of~$A$ on any vector.

    Let $Z_1 = \sigma^z \otimes \one_{2}^{\otimes (\nn -1)} \in \pauli_\nn$.
    Let $\ket \psi$ be any vector such that $J_{Z_1} \ket \psi = \ell \ket \psi$,
    where $\ell$ is the greatest eigenvalue.
    There are $4^{\nn -1}$ unordered pairs $\{\sigma^x \otimes W, \sigma^y \otimes W\}$ 
    where $W \in \{\one_2, \sigma^x, \sigma^y, \sigma^z\}^{\otimes (\nn-1)}$
    such that $\CC \{Z_1, \sigma^x \otimes W, \sigma^y \otimes W \} \cong \sutwo$
    as Lie algebras.
    We know if $X,Y,Z$ are a triple generating~$\sutwo$ such that $[X,Y] = \ii Z$ 
    (and cyclic permutations thereof),
    then
    \begin{equation}
        \rho_*(X)^2 + \rho_*(Y)^2 + \rho_*(Z)^2 = \ell_\rho(\ell_\rho+1) \one
    \end{equation}
    for any irrep~$\rho_*$ where $\ell_\rho$ is the greatest eigenvalue of~$\rho_*(Z)$.
    The linear span of all vectors obtained by acting with $J_{Z_1}, J_{\sigma^x \otimes W}, J_{\sigma^y \otimes W}$ on~$\ket \psi$
    is an $\sutwo$-irrep because $J_{Z_1}$ assumes the greatest eigenvalue~$\ell$ on~$\ket \psi$, 
    and hence
    \begin{align}
        (J_{Z_1}^2 +J_{\sigma^x \otimes W}^2 + J_{\sigma^y \otimes W}^2)\ket \psi &= (\ell^2 + \ell)\ket \psi,  \nonumber\\
        (J_{\sigma^x \otimes W}^2 + J_{\sigma^y \otimes W}^2)\ket \psi &= \ell \ket \psi.
    \end{align}
    Therefore,
    \begin{equation}
        \bra \psi A \ket \psi 
        \ge 
        \bra \psi J_{Z_1}^2 \ket \psi
        +
        \sum_{W} \bra \psi (J_{\sigma^x \otimes W}^2 + J_{\sigma^y \otimes W}^2) \ket \psi 
        = \ell^2 + 4^{\nn -1} \ell .
    \end{equation}
    This proves the lower bound.
\end{proof}

The remainder of the proof uses highest weights.

\begin{proof}[Proof of the rest of the claims in~\cref{thm:QuadraticCasimirBound}]
    To prove the upper bound we take $\ket \psi$ to be a highest weight vector.
    By definition, this means that $\ket \psi$ is annihilated by all positive roots of~$\su(\NN)$,
    which span the $\CC$-linear space of all strictly upper triangular $\NN$-by-$\NN$ matrices.
    This means that
    \begin{equation}
        (J_X + \ii J_Y) \ket \psi = 0
    \end{equation}
    for any $X + \ii Y \in \su(\NN)$ that is upper triangular in a standard basis for Pauli operators.
    Then,
    \begin{equation}
        0 = \bra \psi ( J_X - \ii J_Y) (J_X +\ii J_Y) \ket \psi = \bra \psi (J_X^2 + J_Y^2 - J_Z) \ket \psi \label{eq:JxJyJz}
    \end{equation}
    where $\ii Z = [X,Y]$. It follows that
    \begin{equation}
        \bra \psi J_X^2 + J_Y^2 \ket \psi = \bra \psi J_Z \ket \psi \le \ell . \label{eq:JzLeEll}
    \end{equation}
    To use this we partition~$\pauli_n$ as follows.
    For a bitstring~$z\in \{0,1\}^{k-1}$ of length~$k-1$,
    define $Z(z) = \bigotimes_{j=1}^{k-1} (\sigma^z)^{z_j} \in \pauli_{k-1}$.
    For a $\ZZ_4$-string~$w \in \{0,1,2,3\}^{\nn - k}$,
    define $W(w) = \bigotimes_{j=1}^{\nn - k} \sigma^{w_j}$
    where $\sigma^0 = \one_2$, $\sigma^1 = \sigma^x$, $\sigma^2=\sigma^y$, and $\sigma^3 = \sigma^z$.
    Then, a triple
    \begin{equation}
            Z(z) \otimes \sigma^x \otimes W(w), 
            \quad Z(z) \otimes \sigma^y \otimes W(w), \quad
            \one_2^{\otimes (k-1)} \otimes \sigma^z \otimes \one_2^{\otimes (\nn -k)}
    \end{equation}
    forms an $\sutwo$ subalgebra.
    So we have identified $\sum_{k=1}^{\nn} 2^{k-1} 4^{\nn - k} = 4^\nn(2^{-1} - 2^{-\nn-1})$
    subalgebras.
    The pairs $\{Z(z) \otimes \sigma^x \otimes W(w),  Z(z) \otimes \sigma^y \otimes W(w)\}$ 
    account for $4^\nn - 2^\nn$ elements of~$\pauli_\nn$,
    and the remaining $2^\nn-1$ operators of~$\pauli_\nn$ are tensor products of~$\one_2$ and~$\sigma^z$.
    Therefore,
    \begin{equation}
        \sum_{P \in \pauli_\nn} \bra \psi J_P^2 \ket \psi \le 4^\nn(2^{-1} - 2^{-\nn-1})\ell + (2^\nn -1) \ell^2 .
    \end{equation}

    The tightness of the bounds can be shown by considering 
    specific highest weights and using the partition above.
    The upper bound is saturated if and only if 
    \cref{eq:JzLeEll} is saturated for all nonidentity tensor product~$Z$ of~$\one_2$ and $\sigma^z$.
    So, we need to show that such a linear functional on the Cartan subalgebra
    is a valid point on the weight lattice.
    This is easy: if $L_j$ denotes the dual vector of the diagonal matrix 
    with a sole $1$ at the $j$-th diagonal, $2\ell L_1$
    is the desired weight.
    
    To prove that the lower bound can be saturated,
    we take a highest weight in which only $J_{Z_1}$, where $Z_1 = \sigma^z \otimes \one_2^{\otimes (\nn-1)}$, 
    takes the greatest eigenvalue~$\ell$, but $J_Z$ for any other diagonal~$Z$ assumes zero.
    This amounts to the weight~$k \sum_{j=1}^{\NN/2} L_j $, giving $\ell = k\NN/4$ for some positive integer~$k$.
    Then, on the highest weight vector~$\ket \psi$,
    the generator~$Z_1$ gives~$\ell^2$, 
    and the $4^{\nn -1}$ pairs $\{ \sigma^x \otimes W, \sigma^y \otimes W\}$ gives~$4^{\nn -1} \ell$,
    but all other $2^\nn - 2$ diagonal generators give zero by the choice of the highest weight.
    Finally, \cref{eq:JxJyJz} implies that the $4^\nn - 2^\nn - 2\cdot 4^{\nn -1}$ 
    generators $Z(z) \otimes \sigma^x \otimes W(w),  Z(z) \otimes \sigma^y \otimes W(w)$
    give zero. 
    
    In those saturating conditions,
    we are forced to choose a unique highest weight given~$\ell$,
    which in turn determines the irrep (up to equivalence).
\end{proof}

Next, we present an alternative and slightly tighter bound, 
which assumes familiarity with highest weights for Lie algebra representations (e.g.~\cite[\S 14-15]{FultonHarris}). 
Note, again, that our proof of the lower bound in \cref{thm:QuadraticCasimirBound} only uses the representation theory of $\su(2)$.

    A finite-dimensional irreducible representation of $\su(\NN)$ is labeled by its highest weight $\sum_j \mu_j L_j$, 
    which is labeled by a sequence of $\NN$ integers $\mu=(\mu_1,\mu_2,\dots,\mu_\NN)$ 
    where $\mu_i \ge \mu_{i+1}$ modulo integer multiples of~$(1,1,\ldots,1)$. 
    We choose a representative~$\mu$ such that $\sum_{i=1}^N \mu_i=0$, 
    which is possible for subrepresentations of~$\tau : U \mapsto (U \otimes \bar U)^{\otimes t}$,
    and henceforth the Killing form~$\langle \cdot, \cdot \rangle$ on the dual of Cartan subalgebra 
    is given by~$\langle \mu, \mu' \rangle = \sum_i \mu_i \mu'_i$.
    (The normalization here is different from that in~\cite[\S15]{FultonHarris}.)
    
    Let $H=\half\sigma^z\otimes \one_2^{\otimes \nn-1}=\frac{1}{2}\diag(1,1,\dots,1,-1,-1,\dots,-1)$ 
    be an element of the Cartan subalgebra. We would like to determine $\ell=\norm{\phi_*(H)}$. 
    Since the set of weights are in the convex hull of the Weyl group orbit of $\mu$, 
    the maximum eigenvalue of $\phi_*(H)$ is given by 
    $\ell=\max_{w\in \mathrm{Weyl}} (w\cdot\mu)(H) = \mu(H) = 
    \half \left(\mu_1+\cdots +\mu_{\NN/2} - \mu_{\NN/2 + 1} - \cdots - \mu_{\NN}\right)$.

    Next, we invoke a formula for the quadratic Casimir operator (e.g.~\cite[(25.14)]{FultonHarris}): 
    For any finite-dimensional irrep $\phi_*$ with highest weight $\mu$,
    and for any basis $\{X_j\}$ of $\su(\NN)$ that is orthonormal with respect to the Killing form,
    \begin{equation}
        \sum_{j}\phi_*(X_j)^2=\left(\langle\mu,\mu\rangle+\langle\mu, \delta\rangle\right)\one,
    \end{equation}
    where $\delta$ is the sum of all positive roots, which can be written in a vector form as 
    $\delta = (\NN-1,\NN-3,\dots,-(\NN-1))$.
    Therefore, $\frac{4}{\NN}\sum_{P \in \pauli_\nn} J_P^2 = \left(\langle\mu,\mu\rangle+\langle\mu,\delta\rangle\right)\one$, 
    where the factor $4/\NN$ comes from renormalizing $J_P$ to an orthonormal basis
    ($\Tr((P/2)^2) = \NN/4$).
    Now, our quantity of interest is exactly determined by the highest weight of a given irrep:
    \begin{equation}\label{eq:casimircalculation}
        \frac{1}{\ell^2}\sum_{P \in \pauli_\nn} J_P^2
        =
        \frac{\NN\left(\langle\mu,\mu\rangle+\langle\mu,\delta\rangle\right)}{\left(\mu_1+\cdots +\mu_{\NN/2}-\mu_{\NN/2 + 1}-\cdots -\mu_{\NN}\right)^2}\one.
    \end{equation}

\begin{lemma}\label{thm:improvedlowerbound}
    Assume $t \le \NN/2$. 
    Let $\phi_*$ be an irreducible $\su(\NN)$-subrepresentation of~$\tau: U \mapsto (U \otimes \bar U)^{\otimes t}$.
    Let $\ell = \norm{\phi_*(P/2)}$ be the Schatten $\infty$-norm, which is independent of~$P \in \pauli_\nn$.
    \begin{equation}
        \frac{1}{\ell^2}\sum_{P \in \pauli_\nn} J_P^2 \succeq \frac{\NN(\NN-t+1)}{2t}\one.
    \end{equation}
    There exists an irreducible subrepresentation~$\phi_*$ of~$\tau$ achieving the equality.
\end{lemma}

Roughly speaking, in the small-$t$ regime where $t\ll\NN$, 
this gives a factor of 2 improvement relative to \cref{thm:QuadraticCasimirBound}.
Note that this is consistent with the tightness stated in \cref{thm:QuadraticCasimirBound},
because the irreps that saturate the lower bound of \cref{thm:QuadraticCasimirBound} are not subrepresentations of $\tau$ when $t\ll\NN$.
If $t = \NN/2$, then this lemma gives the same bound as~\cref{thm:QuadraticCasimirBound}.

The strategy in the proof below applies for $t > \NN/2$
and gives an alternative proof of the lower bound of \cref{thm:QuadraticCasimirBound},
but we will omit such calculation.

\begin{proof}
    It only remains to perform an elementary calculation: 
    minimize the right-hand side of \cref{eq:casimircalculation} 
    over the highest weights that correspond to irreducible subrepresentations of~$\tau$.

    The decomposition of $\tau$ into irreps is well understood (e.g.~\cite[Theorem 4]{Roy2009}).
    An irrep with highest weight $\mu$ is a subrepresentation of $\tau$ 
    if and only if $\sum_i \mu_i = 0$ and $\sum_i \abs{\mu_i} \leq 2t$.
    When $t\leq\NN/2$, the minimizer is given by
    \begin{equation}
    \mu_*=\left(1,1,\dots,1,0,0,\dots,0,-1,-1,\dots,-1\right)
    \end{equation}
    with the first $t$ entries being 1 and last $t$ entries being $-1$. 
    Calculating the right-hand side of \cref{eq:casimircalculation} with respect to $\mu_*$ gives the stated bound.
\end{proof}

\begin{remark}
    If~\cref{eq:avgkernelprojbound} is saturated, then~$t \in \{1,2,4\}$.
    If $\ell \notin \{1,2,4\}$,
    then we know by comparing~\cref{thm:sutwoSpectralGap} with~\cref{eq:maingapbound} that
    $K(J_x) + K(J_y) + K(J_z) + (J_x^2 + J_y^2 + J_z^2) / \ell^2$ has norm strictly smaller than~$3$
    in any nontrivial irrep of~$\su(\NN)$ with $\norm{J_x} = \ell$
    for any triple $J_x, J_y, J_z$ that form an $\sutwo$ subalgebra.
\end{remark}

\begin{remark}
    It may be significant underestimation of the spectral gap 
    using $K(H) \preceq \one -  H^2 / \norm{H}^2$.
    Take the case of~$\su(2)$ and consider a random Pauli rotation by~$\{(\sigma^x, \half), (\sigma^y, \half)\}$.
    In the $(2\ell+1)$-dimensional $\sutwo$-irrep, 
    since $J_x^2 + J_y^2 + J_z^2 = (\ell^2 + \ell)\one$ and $J_z^2 \preceq \ell^2 \one$,
    we see $(J_x^2 + J_y^2)/\ell^2 \succeq \one / \ell$,
    which is best possible 
    since $\bra \psi (J_x^2 + J_y^2) \ket \psi /\ell^2 = 1/\ell$ 
    if $J_z \ket \psi = \ell \ket \psi$.
    This gives a lower bound on the spectral gap~$\Omega(1/t)$. 
    However, by the exact calculation in~\cref{thm:sutwoKOverlap} above,
    we know that the spectral gap of this design is independent of~$t$.
\end{remark}

\begin{proof}[Proof of~\cref{thm:main1}]
    The spectral gap bound is proved in~\cref{thm:SpectralGapByQuadraticCasimir}.
    Note that a random Pauli rotation is defined by a uniform probability distribution on~$\pauli_\nn \times (-\pi,\pi)$ 
    while the second statement of~\cref{thm:main1} takes a uniform distribution on~$\pauli_\nn \times (-2\pi,2\pi)$.
    
    The discrepancy of a factor of~$2$ here is completely dismissable for the first statement of~\cref{thm:main1}
    since we only need a range of~$\theta$ in~$\exp(\ii \theta \phi_*(P/2))$ 
    such that the average over~$\theta$ of 
    a represented operator~$\phi_*(P/2)$ for any Pauli operator~$P \in \pauli_\nn$
    is the projection onto the kernel of~$\phi_*(P/2)$;
    for irreps~$\phi$ that appear in~$\tau : U \mapsto (U \otimes \bar U)^{\otimes t}$,
    the eigenvalues of~$\phi_*(P/2)$ are integers.
    
    However, if we consider all finite dimensional unitary representation~$\rho$ of~$\SU(\NN)$,
    we are no more guaranteed that $\rho_*(P/2)$ has integer eigenvalues.
    Fortunately, every eigenvalue of~$\rho_*(P/2)$ is half an integer for any unitary representation~$\rho$
    because every finite dimensional unitary irrep of~$\SU(\NN)$
    is a subrepresentation of~$U \mapsto U^{\otimes m}$ 
    for some integer~$m \ge 0$~\cite[\S15.3]{FultonHarris}.
    
    Note that for an irrep~$\rho_*$ where $2\norm{\rho_*(P/2)}$ is an odd integer,
    the kernel of~$\rho_*(P/2)$ is zero, implying that averaging over $\theta$ eliminates this irrep.
\end{proof}

\begin{proof}[Proof of~\cref{thm:main2}]
    Let $\calD = \calC_t^k - \calH_t$ be the difference of the channels.
    With~$k \geq (4\log 2)\nn t^2 + 4t \log \frac 1 \eps$,
    we have $\norm{\calD \otimes \mathcal I}_{2\to 2} = \norm{\calD}_{2 \to 2} \le (1 - \frac 1 {4t})^k \le \eps 2^{-\nn t}$
    by~\cref{thm:main1},
    where $\mathcal I$ means the identity channel on any auxiliary system.
    Since the diamond norm is obtained by taking an equal dimensional auxiliary system,
    $\norm{\calD}_\diamond = \norm{\calD \otimes \mathcal I_{2^{\nn t}}}_{1 \to 1} \le 2^{\nn t} \norm{\calD \otimes \mathcal I_{2^{\nn t}}}_{2 \to 2} \le \eps$,
    where the first equality is by~\cite[Theorem~11.1]{KSV}.
    (The use of $2 \to 2$ norm for this purpose has appeared in~\cite{Low2010thesis,Brand_o_2016}.)
    
    Similarly, if~$k \geq (4\log 8)\nn t^2 + 4t \log \frac 1 \eps$,
    we have $\norm*{\calC^k_t - \calH_t}_\diamond \le \eps 4^{-\nn t}$ 
    which implies $(1-\eps)\calH_t\preceq  \calC_t^k\preceq(1+\eps)\calH_t$ 
    by~\cite[Lemma 3]{Brand_o_2016}. 
    
    The gate complexity follows from~\cref{thm:singlePauliRotationComplexity}.
\end{proof}

\section{Orthogonal designs}\label{sec:orthogonaldesigns}

Our approach can be adapted for special orthogonal groups.
This section uses arguments parallel to those in the analysis for~$\SU(\NN)$,
so we will be rather brief.

\subsection{Skew-symmetric Pauli operators}

We consider the special orthogonal group~$\SO(\NN)$ in a fashion similar to our random Pauli rotations.
We directly use the inclusion $\SO(\NN) = \SU(\NN) \cap \RR^{\NN \times \NN} \subset \SU(\NN)$ for $\NN = 2^\nn$.

\newcommand{\skewpauli}{{\mathbf{Y}}}

Define a set~$\skewpauli_\nn$ of Pauli operators with entries in~$\ii \RR$:
\begin{equation}
    \skewpauli_\nn = \{ P \in \pauli_\nn \,|\, \text{An odd number of $\sigma^y$ tensor factors appear in~$P$.} \}
\end{equation}

We first verify that the $\RR$-linear span of~$\ii\skewpauli_\nn$ is precisely the real Lie algebra~$\so(\NN = 2^\nn)$
consisting of all antisymmetric real matrices.
It is clear that $\ii \skewpauli_\nn$ is $\RR$-linearly independent and consists of skew-symmetric real matrices.
We can check that $\abs{\skewpauli_\nn} = \NN (\NN-1)/2$ by solving a recursion equation as follows.
Let $e(\nn)$ be the number of all Pauli operators $\{ \one, \sigma^x, \sigma^y, \sigma^z \}^{\otimes \nn}$ 
that contain an even number of tensor factors~$\sigma^y$.
The identity operator~$\one_2^{\otimes \nn}$ contributes~$1$ to~$e(\nn)$.
Consider a subset of Pauli ``strings'' whose first ``letter'' is $\sigma^y$, 
and another set of Pauli strings whose first letter is one of~$\one,\sigma^x,\sigma^z$.
It is then clear that $\abs{\skewpauli_{\nn+1}} = 3 \abs{\skewpauli_\nn} + e(\nn)$ 
and $e(\nn+1) = 3 e(\nn) + \abs{\skewpauli_\nn}$ with initial conditions~$e(1) = 3$ and $\abs{\skewpauli_1} = 1$.
The claim~$\abs{\skewpauli_\nn} = 2^{\nn-1} (2^\nn -1)$ follows by induction in~$\nn$.

\begin{theorem}\label{thm:RandomSkewSymmetricPauliRotation}
    Suppose $\NN = 2^\nn > 4$.
    For any integer~$t \ge 1$, we have
    \begin{equation}
    \norm*{
        \E_{\theta \sim (-2\pi,2\pi), \, P \in \skewpauli_\nn} (e^{\ii \theta P / 2})^{\otimes t} - \E_{O \sim \SO(\NN)} O^{\otimes t}
    }
    \le
    1 - \frac{1}{2t}\frac{\NN - 2}{\NN - 1} - \frac{2}{\NN(\NN-1)} .
    \end{equation}
\end{theorem}

The small orthogonal groups, $\SO(2)$ and $\SO(4)$, 
are excluded for simplicity of the proof as they are not simple Lie groups.
Note that every finite dimensional unitary irrep\footnote{%
    A representation of the Lie group~$\SO(\NN)$ induces a representation of its Lie algebra~$\so(\NN)$, 
    but not every representation of~$\so(\NN)$ arises in this way,
    and the seed for the remaining $\so(\NN)$-irreps is spin representations.
    This phenomenon does not happen for~$\SU(\NN)$.
    It does not seem to make sense to consider spin representations in orthogonal designs
    since the Haar average over~$\SO(\NN)$ of a spin representation is ill-defined.
}
of~$\SO(\NN)$ for $\NN > 4$ appears as a subrepresentation of the tensor representation~$O \mapsto O^{\otimes t}$ 
for some integer~$t \ge 1$.
This is explained in~\cite[\S19]{FultonHarris}.
Therefore, \cref{thm:RandomSkewSymmetricPauliRotation} implies that
for any irrep of~$\SO(\NN = 2^\nn > 4)$ the spectral gap is at least $1 / \dim \SO(\NN)$.

A representation of~$\SU(\NN)$ gives a representation of~$\SO(\NN)$,
but a representation of~$\SO(\NN)$ does not in general give a representation of~$\SU(\NN)$.
So, this theorem cannot be thought of as a corollary of~\cref{thm:main1}.
The proof below is however almost identical to that of~\cref{thm:main1},
mainly because we use only the common aspects of the representation theories of~$\SO(\NN)$ and $\SU(\NN)$.
Thus, we will assume a reader's familiarity with the proof of~\cref{thm:main1},
or rather the proofs of~\cref{thm:SpectralGapByQuadraticCasimir,thm:QuadraticCasimirBound}, and omit some detail.

\begin{proof}
    (Step~0: to irreps)
    The Haar average (the second term in the norm) 
    is the projector onto the trivial subrepresentation (\cref{thm:HaarProjectionTrivialRep}).
    Hence, we consider an irreducible nontrivial $\so(\NN)$-subrepresentation~$\rho_*$ of~$O \mapsto O^{\otimes t}$.
    
    (Step~1: random angles give kernel projectors.)
    It is clear that $\rho_*(P/2)$ has half-integer eigenspectrum for any~$P$,
    and therefore the average over~$\theta \in (-2\pi,2\pi)$
    eliminates all nonzero eigenvalues: $\E_{\theta\sim(-2\pi,2\pi)} e^{\ii \theta \rho_*(P / 2)} = K(\rho_*(P/2))$
    is the projector onto the kernel of~$\rho_*(P/2)$.\footnote{
        We implicitly allowed complex coefficients for~$\so(\NN)$, i.e., the complexification.
        Exponentiated matrices are all real.
    }
    
    (Step~2: identical eigenspectra for all represented operators)
    The norm of~$\rho_*(P/2)$ is independent of~$P \in \skewpauli_\nn$:
    this is an analog of~\cref{thm:EqualSpectrum} for~$\so(\NN)$,
    and the proof is similar.
    Note that for any matrix~$M$, we have $\det (M \otimes \one_2) = \det (M \oplus M) = (\det M)^2$.
    So, if $O \in \mathsf O(\NN/2)$, then $O \otimes \one_2 \in \SO(\NN)$.
    The Clifford unitaries, $\mathrm{CNOT}$ and Hadamard, are in~$O(4)$.
    If there is a tensor factor~$\sigma^y \otimes \sigma^y$
    in some $P \in \skewpauli_\nn$,
    then we must have $\nn \ge 3$ since $P$ must contain an odd number of~$\sigma^y$'s.
    Hence, using $\mathrm{CNOT}$ and Hadamard Clifford unitary acting on those two $2 \times 2$ tensor factors,
    we can turn~$P$ by $\SO(\NN)$ conjugation into a Pauli that where $\sigma^y \otimes \sigma^y$ is replaced
    by~$\sigma^x \otimes \sigma^x$ while not changing any other tensor factor of~$P$. 
    Inductively, we turn all pairs of~$\sigma^y$ tensor factors into pairs of~$\sigma^x$.
    By the same argument, we can turn any $\sigma^z$ tensor factor into $\sigma^x$.
    Now, the conjugation of $\sigma^y \otimes \one_2$ by~$\mathrm{CNOT}$ is $\sigma^y \otimes \sigma^x$.
    Therefore, any $P \in \skewpauli_\nn$ with $\nn \geq 3$ is congruent to $\sigma^y \otimes \one_2^{\otimes (\nn -1)}$
    by some element of~$\SO(\NN)$.
    Hence for any $P \in \skewpauli_\nn$ where $\nn \ge 3$, there exists $O \in \SO(\NN)$ such that 
    $\rho(O)\rho_*(P)\rho(O)^{-1} = \rho_*(\sigma^y \otimes \one_2^{\otimes (\nn-1)})$.
    Put $\ell = \norm{\rho_*(P/2)} \le t/2$ for any~$P \in \skewpauli_\nn$.

    (Step~3: to quadratic Casimir)
    Bounding the kernel projector by the quadratic operator (\cref{eq:opineq}), 
    we are left with the problem of lower bounding
    \begin{equation}
        \E_{P \in \skewpauli_\nn} \rho_*(P/2)^2 / \ell^2 .
    \end{equation}
    By a completely analogous calculation as in~\cref{eq:sutwoCasimir}, this average is a scalar multiple of the identity.
    
    (Step~4: find a large number of~$\sutwo$'s)
    To estimate the unique eigenvalue of this average,
    we look at a vector $\ket \psi$ such that $\rho_*(\half \sigma^y \otimes \one_2^{\otimes (\nn - 1)}) \ket \psi = \ell \ket \psi$.
    We can find $\abs{\skewpauli_{\nn-1}}$ $\sutwo$-subalgebras:
    \begin{equation}
        \sigma^y \otimes \one_2^{\otimes (\nn - 1)},\quad  \sigma^x \otimes W,  \quad \sigma^z \otimes W,
    \end{equation}
    where each triple of $\sutwo$ generators is uniquely labeled by~$W \in \skewpauli_{\nn-1}$.
    Hence,
    \begin{equation}
        \sum_{P \in \skewpauli_\nn} \rho_*(P/2)^2  \succeq \one \left( \ell^2 + \abs{\skewpauli_{\nn-1}} \ell \right) .
    \end{equation}
    The theorem is proved since $\ell \le t/2$.
\end{proof}

\subsection{Skew-symmetric elementary matrix basis}

We give another orthogonal design.
In this subsection, we will not require $\NN$ to be a power of~$2$.

Let $\NN \ge 3$ be any integer.
For any integers $a,b$ ($1 \le a, b \le \NN$),
let $E_{a,b} = \ket a \!\bra b - \ket b \! \bra a$ 
denote the skew-symmetric $\NN \times \NN$ matrix
in which there are only two nonzero matrix entries~$\pm 1$.
Define
\begin{equation}
\skewe_\NN  = \left\{
    E_{a,b}  \in \RR^{\NN \times \NN} \,\,\middle|\,\,
    1 \le a < b \le \NN
    \right\}
\end{equation}
Clearly, $\skewe_\NN$ is a linear basis for~$\so(\NN)$.
We see that $[E_{a,b}, E_{b,c}] = E_{a,c}$ for any~$a,b,c$.
This basis is convenient because different elements are orthogonal with respect to the Killing form.

\begin{theorem}\label{thm:RandomSkewSymmetricElementaryRotation}
    Let $\NN \ge 3$ and $t \ge 1$ be any integers.
    Then,
    \begin{equation}
    \norm*{
        \E_{\theta \sim (-\pi,\pi), \, E \in \skewe_\NN} (e^{\theta E})^{\otimes t} - \E_{O \sim \SO(\NN)} O^{\otimes t}
    }
    \le
    1 - \frac{1}{t}\frac{2(\NN - 2)}{\NN(\NN - 1)} - \frac{2}{\NN(\NN-1)} .
    \end{equation}
\end{theorem}

This can be used to generate an approximately Haar random $\NN \times \NN$ orthogonal matrix fast.
The exponential of an element~$E \in \skewe_\NN$ is a $2\times 2$ matrix, 
direct summed with an $\NN-2$ dimensional identity matrix.
Hence, multiplying a dense $\NN \times \NN$ matrix by $e^{\theta E}$ takes
$\bigO \NN$ arithmetic operations.
For some applications, 
this method can be better than generating $\NN \times \NN$ Gaussian random entries
and running the Gram--Schmidt process.

\begin{proof}
    As before, we consider an $\SO(\NN)$-irrep $\rho$,
    and estimate the norm of~$\E_{\theta,E} \rho(e^{\theta E})$. 
    
    (Step~1: random angles give kernel projectors.)
    The eigenvalues of any $E\in \skewe_\NN$ are $0, \pm \ii$.
    So, the average over~$\theta$ gives $\E_{\theta,E} \rho(e^{\theta E}) = \E_{E} K(\rho_*(E))$.
            
    (Step~2: identical eigenspectra for all represented operators)
    It is obvious that two different elements of~$\skewe_\NN$ are related by some row and column permutations.
    A transposition is not in~$\SO(\NN)$, but the product of a transposition and a diagonal matrix with $\NN-1$ entries being~$1$ and the remaining being~$-1$ is.
    Since $\NN \ge 3$, we can always find such diagonal matrix that leaves a given $E_{a,b}$ fixed 
    --- just look at a zero column or a row.
    Hence, any two elements of~$\skewe_\NN$ are congruent by~$\SO(\NN)$,
    and the eigenspectrum of represented operators~$\rho_*(E)$ is independent of~$E \in \skewe_\NN$.
    Let $\ell = \norm{\rho_*(E)} \le t$ for any~$E \in \skewe_\NN$. 
    
    (Step~3: to quadratic Casimir)
    We need to check that $\sum_{E \in \skewe_\NN} \rho_*(E)^2$ commutes with every~$E \in \skewe_\NN$.
    A moment's thought shows that it suffices to check the commutation of $\rho_*(E_{a,b})$ 
    with $\rho_*(E_{a,c})^2 + \rho_*(E_{b,c})^2$ for any~$c$,
    but this is exactly the same calculation as for the $\sutwo$ case.
    Hence, an upper bound~$\one - \E_E \rho_*(E)^2 / \ell^2$ on $\E_E K(\rho_*(E))$ is a scalar multiple of the identity.

    (Step~4: find a large number of~$\sutwo$'s)
    We focus on a vector with an eigenvalue~$\ell$ for~$\rho_*(E_{1,2})$.
    For any $c \ge 3$, we have an $\sutwo$-subalgebra generated by~$E_{1,2}, E_{1,c}, E_{2,c}$.
    So, $\sum_{E \in \skewe_\NN} \rho_*(E)^2 \succeq \one (\ell^2 + (\NN - 2) \ell)$.
\end{proof}

\subsection*{Acknowledgements}
We thank Thiago Bergamaschi, Jonas Haferkamp, Aram Harrow, Zeph Landau, Ryan O'Donnell, and Peter Shor for helpful discussions. 
This work was done in part while X.T. was visiting the Simons Institute for the Theory of Computing, 
and while J.H., Y.L., and X.T. were at the Park City Mathematics Institute 2023 Graduate Summer School. 
Y.L.~is supported by DOE Grant No. DE-SC0024124, NSF Grant No. 2311733, and MURI Grant No. S394857. 
X.T.~is supported by NSF Grant No. CCF-1729369.

\appendix

\section{Discrete angles}\label{sec:discreteangles}

In the proof of~\cref{thm:main1}, the only place where we use averaging over $\theta\sim(-\pi,\pi)$ is in the following context: for a nonzero hermitian operator $H$ with integer eigenvalues in $[-t,t]$, we have
\begin{equation}
    \E_{\theta\sim(-\pi,\pi)}e^{\ii \theta H}=K(H)\preceq \one-\frac{H^2}{\norm{H}^2},
\end{equation}
where $K(H)$ is the orthogonal projector onto the kernel of $H$ (see~\cref{eq:kernelupperbound}). 
Since $H$ has bounded norm, 
we can instead consider averaging over the angles in the discrete set $\Theta_t=\{m\pi/t:m\in\mathbb{Z}\cap [-t,t-1]\}$. 
Then, for any hermitian operator $H$ with integer eigenvalues in $[-t,t]$, we have
\begin{equation}\label{eq:discreteanglebound}
    \E_{\theta\sim\Theta_t}e^{\ii \theta H}=\E_{\theta\sim(-\pi,\pi)}e^{\ii \theta H}=K(H)
\end{equation}
because for any integer $k$,
\begin{equation}
    \sum_{m=-t}^{t-1} \exp(\ii m k \pi/t)=\begin{cases}
        2t&\text{if }k=0,\\
        0&\text{if }0<|k|\leq t.
    \end{cases}
\end{equation}
The rest of the proof is exactly the same as the proof of~\cref{thm:main1}.

\section{State designs}\label{sec:statedesign}

We consider distributions~$\nu$ on a complex projective space $\mathbb{CP}^{\NN-1}$ 
where $\NN = 2^\nn$ is a power of~$2$.
This is often called a state design because $\mathbb{CP}^{\NN-1}$
is the set of all normalized state vectors modulo global phase factors in an $\nn$-qubit system,
or equivalently the set of all rank-1 projectors $\ket \psi\! \bra \psi$ on~$(\CC^2)^{\otimes \nn}$.
There is a natural (left) action of a unitary group 
given by $\ket \psi \! \bra \psi \mapsto U \ket \psi \! \bra \psi U^\dag$
for $U \in \SU(\NN)$.
The Haar measure of~$\SU(\NN)$ induces an $\SU(\NN)$-invariant measure on~$\mathbb{CP}^{\NN-1}$.
This is the target distribution we wish to approximate.
A natural metric to measure the quality of approximation is 
closeness in $t$-th moments, maximized over all possible measurements.
This is succinctly described by the trace distance:
\begin{equation}
    \frac 1 2 
    \norm[\Big]{\underbrace{\E_{\psi \sim \nu} (\ket \psi \! \bra \psi) ^{\otimes t}}_{\calS_{\nu,t}} - \underbrace{\E_{U \sim \SU(\NN)} (U \ket{\alpha}\!\bra{\alpha}U^\dag)^{\otimes t}}_{\calS_{\mathrm{Haar},t}}  }_1 
    \label{eq:state-distance}
\end{equation}
where $\ket{\alpha}$ can be any normalized vector in~$(\CC^2)^{\otimes \nn}$ due to the right invariance of the Haar measure.
Any approximate unitary design can be used for state designs,
and a bound on the $t$-th moment trace distance directly comes from the analysis of the approximate unitary design.
For example, the result of~\cref{thm:main2} serves the purpose.
However, this is not necessarily the best one can show.

\begin{theorem}\label{thm:statedesign}
    Let $\norm{\cdot}_1$ denote the Schatten $1$-norm of a matrix, the sum of all singular values.
    For any integers $t,k,\nn \ge 1$ and a normalized vector~$\ket \alpha \in \CC^\NN$, 
    we have
    \begin{equation}
        \norm[\Big]{
           \calC_t^k( \ket{\alpha}^{\otimes t}\!\bra{\alpha}^{\otimes t} )
            - 
            \calS_{\mathrm{Haar},t}
        }_1 
        \le
        \binom{\NN + t -1}{t}^{1/2}
        \left( 
            1 - \frac{1}{2 t} \frac{\NN}{\NN+1} - \frac{\NN}{2(\NN^2-1)}
        \right)^k. \label{eq:statedesignresult}
    \end{equation}
\end{theorem}

The last term in the parenthesis is~$\approx (2\NN)^{-1}$ for large~$\NN$,
which contrasts to the last term~$\approx \NN^{-2}$ in~\cref{thm:main1}.

\begin{proof}
    The input $(\ket{\alpha}\!\bra \alpha)^{\otimes t}$ is invariant under
    tensor factor permutations either on the ket or bra factors.
    The action by $\SU(\NN)$ commutes with this permutation symmetry,
    and hence the input vector is in an $\SU(\NN)$-representation $\Sigma = \Sym^t(\CC^\NN) \otimes \Sym^t((\CC^\NN)^*)$.
    By the Littlewood--Richardson rule (actually its special case~\cite[15.25(i)]{FultonHarris}),
    we have a decomposition of~$\Sigma$ into irreps:
    \begin{equation}
        \Sigma = \Sym^t(\CC^\NN) \otimes \Sym^t((\CC^\NN)^*) 
        =\bigoplus_{s=0}^t \underbrace{\mathrm{highest\,weight}~ s (L_1 - L_\NN) }_{\Gamma_s}
    \end{equation}
    Here, $L_i$ is the dual of the diagonal matrix (an element of the Cartan subalgebra)
    where there is a sole nonzero entry that is~$1$ at the $i$-th position.
    Note that all the multiplicities of the irreps are~$1$,
    and $\Gamma_0$ is a one-dimensional trivial representation.
    Decompose~$(\ket \alpha \! \bra \alpha)^{\otimes t}$ into $\bigoplus_{s=0}^t \gamma_s(\alpha)$
    according to the irrep decomposition~$\Gamma_s$.%
    \footnote{
        Here, $\ket \alpha ^{\otimes t} \! \bra \alpha ^{\otimes t}$ is a vector of the representation space~$\Sigma$. The inner product is inherited from $(\CC^\NN)^{\otimes t} \otimes (\CC^\NN)^{* \otimes t}$ and is thus the Hilbert--Schmidt inner product.
        Since $\Sigma$ is a unitary representation,
        the components $\gamma_s(\alpha) \in \Sigma$ for $s = 0,1,\ldots, t$ are orthogonal to each other,
        and have length (defined by the inner product) 
        equal to the Schatten $2$-norm of the corresponding matrix in~$(\CC^\NN)^{\otimes t} \otimes (\CC^\NN)^{* \otimes t}$.
    }
    \cref{thm:HaarProjectionTrivialRep} applied to~$\Sigma$ says 
    that the Haar average of~$\calS_{\mathrm{Haar},t}$
    projects $(\ket \alpha \! \bra \alpha)^{\otimes t}$ onto~$\gamma_0(\alpha)$.
    This projection is independent of~$\alpha$ because~$\gamma_0 = \gamma_0(\alpha)$ is uniquely
    determined by the trace-preserving property.%
    \footnote{
        The other components $\gamma_s(\alpha)$ with $s > 0$ depend on~$\alpha$,
        but $\norm{\gamma_s(\alpha)}_2$ are independent of~$\alpha$,
        since they are invariant under~$\SU(\NN)$ whose action is transitive 
        on~$\mathbb{CP}^{\NN-1}$.
    }
    
    It is now clear that
    \begin{equation}
        \calC_t^k( (\ket \alpha \! \bra \alpha)^{\otimes t} ) - \calS_{\mathrm{Haar},t}
        =
        \calC_t^k\left( \bigoplus_{s = 1}^t \gamma_s(\alpha) \right). \label{eq:statetarget}
    \end{equation}
    We are going to bound the Schatten $2$-norm of~\cref{eq:statetarget} 
    by the factor in the parenthesis of~\cref{eq:statedesignresult}.
    Since it is after all a matrix acting on~$\Sym^t(\CC^\NN)$ of dimension~$\binom{\NN + t -1}{t}$,
    conversion to the Schatten $1$-norm gives the theorem.

    \Cref{thm:ThetaAverage} says that
    $\norm[\big]{\calC_t^k|_{\Gamma_s}}_{2 \to 2} = \norm{\E_{P \in \pauli_\nn} K(\Gamma_{s*}(P/2))}^k$ where $\Gamma_{s*}$ is the induced Lie algebra representation.
    Invoking~\cref{eq:kernelupperbound}, \cref{thm:quadraticSumIsInvariant}, 
    and most importantly~\cref{eq:casimircalculation} with highest weights $s(L_1 - L_\NN)$
    where $s = 1,2,\ldots,t$,
    we find that
    \begin{equation}
        \norm*{
            \bigoplus_{s=1}^t \E_P K(\Gamma_{s*}(P/2))
        }
        \le
            1 - \frac{1}{2 t} \frac{\NN}{\NN+1} - \frac{\NN}{2(\NN^2-1)} .
        \qedhere
    \end{equation}
\end{proof}

One can perform similar calculation for a real projective space~$\mathbb{RP}^{\NN-1}$,
which is the same as the hemisphere (spherical cap) of dimension~$\NN-1$,
excluding the equator of measure zero. 

\printbibliography

\end{document}